\documentclass[11pt]{article}
\def\comments{0}
\def\conferenceversion{0}

\usepackage{algorithm}
\usepackage[noend]{algpseudocode}
\usepackage{amsfonts}
\usepackage{amsmath}
\usepackage{amsthm}
\usepackage{color}
\usepackage[svgnames]{xcolor}

\newtheorem{theorem}{Theorem}[section]

\newtheorem{lemma}[theorem]{Lemma}
\newtheorem{corollary}[theorem]{Corollary}
\newtheorem{definition}[theorem]{Definition}

\newtheorem{remark}[theorem]{Remark}

\newtheorem*{claim*}{Claim}

\definecolor{myred}{rgb}{0.64, 0.0, 0.0}
\definecolor{myblue}{rgb}{0.25,0.50,0.9}
\usepackage[%
backref=page,
colorlinks=true,
urlcolor=myblue,
citecolor=myred,
linkcolor=myred
]%
{hyperref}
\usepackage[capitalize]{cleveref}

\usepackage{verbatim}
\usepackage[makeroom]{cancel}
\usepackage{rotating}
\usepackage{float}
\usepackage{fullpage}
\usepackage{amssymb}
\usepackage{stmaryrd}
\usepackage{mathrsfs}
\usepackage{epsfig}
\usepackage{indentfirst}
\usepackage{framed}
\usepackage[numbers]{natbib}
\usepackage{color}
\usepackage{graphicx}
\usepackage{float}
\usepackage{bbm}
\usepackage{mathtools}
\usepackage{enumitem}
\usepackage{thmtools}
\usepackage{thm-restate}

\usepackage{./shortcuts-full}
\usepackage{./shortcuts_alphabet}

\usepackage{setspace}
\usepackage[explicit]{titlesec}
\usepackage[left=1.2in, right=1.2in]{geometry}

\titleformat{\subsubsection}[runin]{\normalfont\bfseries}{\thesubsubsection.}{1em}{#1.}
\titleformat{\paragraph}[runin]{\normalfont\bfseries}{.}{1em}{#1}
\titleformat{\subparagraph}[runin]{\normalfont\itshape}{.}{0pt}{#1}
\titlespacing*{\subparagraph}{-1pt}{3.25ex plus 1ex minus .2ex}{1em}

\usepackage[textwidth=2.5cm]{todonotes}
\usepackage{bussproofs}
\usepackage{tocloft}
\usepackage{tikz}
\usetikzlibrary{calc}
\usetikzlibrary{decorations.pathreplacing,angles,quotes}
\usetikzlibrary{positioning,matrix,shapes,arrows.meta}
\usepackage{standalone}
\usepackage{tabularray}
\usepackage{longtable}
\usepackage{ragged2e}
\usepackage{makecell}

\newcommand{\myref}[1]{\autoref{#1}}

\title{ETH-Hardness of Learning Monotone Circuits and Approximating Their Size}

\author{%
\makebox[.23\linewidth]{Bruno Cavalar}\\
\textsl{University of Oxford}%
\and
\makebox[.23\linewidth]{Susanna F. de Rezende}\\
\textsl{Lund University}%
\and
\makebox[.23\linewidth]{Matthew Gray}\\
\textsl{University of Oxford}%
\and
\makebox[.23\linewidth]{Rahul Santhanam}\\
\textsl{University of Oxford}%
}

\newcommand{\myparagraph}[1]{\paragraph{#1}}

\begin{document}

\maketitle

\begin{abstract}
    We show the following hardness results for monotone learning and approximation of monotone circuit size:
    \begin{enumerate}
    \item Under the Randomised Exponential-Time Hypothesis (rETH), it requires time $n^{\Omega(\log n)}$ to PAC-learn monotone formulas with $n$ input bits and size $s(n) = n$ by monotone circuits of size $n^{(\log n)^{1-\epsilon}}$, for every $\epsilon > 0$.
    \item Under the Randomised Exponential-Time Hypothesis (rETH), for any $\delta > 0$, there is a polynomially bounded function $m$ such that $m^{1-\delta}$-multiplicatively approximating the minimum monotone circuit size of a monotone function consistent with a sequence of $m(n)$ labelled examples $\{(x_i, b_i)\}$ over $n$-bit inputs requires time $m^{\Omega(\log(m))}$.
    \end{enumerate}
    Our results are shown by a novel application of lifting arguments in proof and communication complexity to hardness of monotone learning, by building on the seminal result of Atserias and M\"uller~\cite{DBLP:journals/jacm/AtseriasM20} on hardness of automating Resolution proofs.
\end{abstract}

\section{Introduction}

\myparagraph{Meta-Complexity} One of the main goals of the area of {\it meta-complexity} is to 
prove strong hardness results for computing complexity measures such as circuit size, formula size 
and monotone circuit size. While such hardness results can be derived in a standard way from 
cryptographic assumptions such as the existence of one-way functions~\cite{RR97, KC00}, it has 
proven much more challenging to show hardness results under {\it worst-case assumptions}, such 
as $\mathsf{NP}$-hardness results. A good reason for this difficulty was identified by Kabanets and 
Cai~\cite{KC00}, who showed that $\mathsf{NP}$-hardness of the Minimum Circuit Size Problem 
(which asks whether a Boolean function F given by its truth table has circuits of size at most $s$, for 
a given parameter $s$) via polynomial-time honest reductions would imply hard-to-prove circuit 
lower bounds such as super-polynomial circuit lower bounds for $\mathsf{EXP}$.

Despite this obstacle, there are now several hardness results for various notions of circuit size under worst-case hardness assumptions, including DNF size, Boolean formula size (under ETH), multi-output circuit size, oracle circuit size, etc.,~\cite{Masek79, AHMPS08, HOS18, ILO20, I20a, I20b, I21, H22, HIR25}. A common feature of these results, except for the recent inapproximability result of~\cite{HIR25} for oracle circuit size, is that they do not show strong inapproximability results for circuit size. Here, by a strong inapproximability result, we mean the hardness of approximating the circuit size within a $N^{\Omega(1)}$ multiplicative factor, where $N$ is the size of the function representation.

A major motivation for showing strong inapproximability result for circuit size under worst-case assumptions is that they offer a pathway to worst-case to average-case reductions for $\mathsf{NP}$. It was shown in~\cite{S20, H23}, building on work of~\cite{CIKK16}, that there are non-black-box reductions from strong inapproximability of problems such as circuit size and average circuit size to average-case hardness for these problems. Thus, $\mathsf{NP}$-hardness of strong inapproximability for either of these problems would rule out Heuristica, the complexity world imagined by Impagliazzo~\cite{Impagliazzo95} where $\mathsf{NP}$ is hard in the worst case but easy on average. A recent result of Mazor and Pass~\cite{MP24} rules out strong inapproximability of circuit size under Levin reductions from SAT, assuming a standard cryptographic assumption, but strong inapproximability might still hold under a more general notion of reduction such as Turing reductions. The strong inapproximability result of~\cite{HIR25} for oracle circuit size is based on cryptographic techniques and seems difficult to apply to standard circuit size measures. It could potentially be very useful to find other techniques yielding strong inapproximability.

In terms of circuit size measures for which hardness under worst-case assumptions is known, 
monotone circuit size has been little-studied despite its central role in complexity theory. Showing 
$\mathsf{NP}$-hardness for the Minimum Monotone Circuit Size Problem (which asks whether a 
monotone Boolean function given by its truth table has monotone circuit size at most $s$, for a 
given parameter $s$) would be a significant step towards $\mathsf{NP}$-hardness for the Minimum 
Circuit Size Problem, as there is a well-known reduction from Minimum Circuit Size to Minimum 
Monotone Circuit Size.

\myparagraph{Learning} One of the main goals of the area of {\it learning theory} is to identify which 
learning problems are hard in the PAC-learning model. When PAC-learning a concept class 
$\mathcal{C}$, we are given samples $(x,f(x))$ from a distribution $\mathcal{D}$, where $f$ 
belongs to $\mathcal{C}$, and our task is to output with high probability a hypothesis $c$ that 
approximates $f$ well according to the distribution. If the output hypothesis $c$ is in $\mathcal{C}$, 
the learning algorithm is said to be {\it proper}, otherwise it is said to be {\it improper}. We also 
consider a {\it semi-proper} version of learning $\mathcal{C}$-circuits by $\mathcal{C'}$-circuits, 
where $\mathcal{C'} \supseteq \mathcal{C}$. 

It is known how to rule out even {\it improper learning} of various concept classes under cryptographic assumptions~\cite{KearnsValiant94, Kharitonov93, DLMSWW09}. More recently, there has been exciting work on ruling out improper learning even of weak concept classes such as DNFs~\cite{Daniely16, DanielyS16, DanielyV21} based on average-case assumptions such as Feige's hypothesis~\cite{Feige02} and its variants. However, just as in the case of meta-complexity problems, showing hardness under worst-case assumptions remains challenging. This was explained in a work of Applebaum, Barak and Xiao~\cite{ABX08}, who showed that $\mathsf{NP}$-hardness of learning general Boolean circuits under Karp reductions or non-adaptive reductions would imply the collapse of the Polynomial Hierarchy. This leaves open the possibility of using more general reductions such as Turing reductions to show hardness.

Known worst-case hardness results are typically for proper learning of weak circuit classes such as DNFs, or for semi-proper learning where the output hypothesis class $\mathcal{D}$ is not too much more powerful than the concept class~\cite{ABFKP08}. While hardness of learning even restricted subclasses of monotone circuits under the uniform distribution is known under cryptographic assumptions~\cite{DLMSWW09}, and unconditional hardness results are known in the statistical query learning framework~\cite{FLS11}, very little seems to be known about the worst-case hardness of learning monotone circuits.

As can be seen from the discussion above, there are several parallels
between the situation with hardness of meta-complexity problems and the
situation with hardness of learning. These parallels are not coincidental:
learning for a circuit class $\mathcal{C}$ can be seen as a {\it search
version} of the problem of determining minimum $\mathcal{C}$-size of a
Boolean function. Indeed, given access to a Boolean function $f$,
outputting a circuit of size $s$ that computes $f$ (perhaps approximately)
via a learning algorithm witnesses that $f$ has (perhaps approximate)
circuit size at most $s$.
This connection is exploited to derive new learning algorithms
in~\cite{CIKK16}.

In this paper, we aim to establish strong results on the hardness of monotone learning based on worst-case assumption, and then to leverage the connection between monotone learning and meta-complexity to get strong inapproximability of monotone circuit size. 

The meta-complexity problem we consider is the following: given $m = \mathsf{poly}(n)$ labelled examples $(x_i, b_i)$, where each $x_i$ is an $n$-bit string and $b_i$ a bit, and a parameter $s$, is there a monotone function $f$ of monotone circuit size at most $s$ that is consistent with the labelled examples, i.e., $b_i = f(x_i)$ for each $i$. In the context of learning, this is known as the Occam Learning problem, and the non-monotone analogue of this problem was called Partial Minimum Circuit Size Problem in~\cite{ILO20}, where it was shown to be $\mathsf{NP}$-hard. We are interested here in strong inapproximability results, which the techniques of~\cite{ILO20} are incapable of proving, since they reduce from the well-approximable Set Cover Problem.

\subsection{Our Results}

Our first main result is on the hardness of PAC-learning monotone formulas. Proving hardness of improperly PAC-learning monotone formulas is at least as hard as showing $\mathsf{NP}$-hardness of learning general Boolean circuits, and hence beyond reach at this point in time. However, we are able to prove a strong result in the {\it semi-proper} setting under worst-case assumptions, where our concept class is the class of monotone formulas, but our hypothesis class is the class of polynomially larger monotone circuits. In this setting, we are able to get a quasipolynomial lower bound on the running time of any PAC-learning algorithm, assuming the randomised Exponential Time Hypothesis. Recall that the randomised Exponential Time Hypothesis (rETH) postulates that any algorithm solving 3SAT on formulas of $n$ variables must take time $2^{\Omega(n)}$. This is a standard hardness assumption in fine-grained complexity and parameterized complexity.

\begin{theorem}[Informal]
\label{thm:introthm1}
Under rETH, for every $\alpha > 0$, monotone formulas of size $n$ on $n$
variables require time $n^{\Omega(\log n)}$ to be PAC-learned by monotone circuits of size
$n^{(\log n)^{1-\alpha}}$.
\end{theorem}

The formal statement of Theorem~\ref{thm:introthm1}
is in the first item of Theorem~\ref{thm:pac-lb} in Section~\ref{sec:improperpac}. To the
best of our knowledge, it was not known previously even how to get hardness of learning
monotone circuits of size $n$ by monotone circuits of size $n^{1+\alpha}$ under standard
worst-case assumptions. In fact, our techniques yield a tradeoff, where we can show
hardness even for learning monotone formulas of even slightly super-polylogarithmic
monotone formula size by monotone circuits of polynomial size. This result is stated in
the second item of Theorem~\ref{thm:pac-lb}.

Our second main result shows the strong inapproximability of monotone circuit size for the partial version of the Minimum Monotone Circuit Size Problem, where we are given a sequence of labelled examples as input. The assumption used here is again rETH. 

\begin{theorem}
\label{thm:introthm2}
For each $c > 1$ the following holds under rETH: there is no algorithm running in time $N^{o(\log N)}$ which, given as input a sequence of $O(n^c \log n)$ labelled examples $(x_i, b_i)$ representable by $N$ bits in all, where each $x_i$ is an $n$-bit string and $b_i$ a bit, can distinguish between the case that these labelled examples are consistent with a monotone formula of size $n$ (i.e., there is $f$ of monotone formula size $n$ for which $f(x_i) = b_i$ for each $i$) and the case that these labelled examples are inconsistent with a monotone circuit of size $n^c$. In particular, for arbitrary $\epsilon > 0$ there is no $N^{o(\log N)}$ time algorithm which $N^{1-\epsilon}$-multiplicatively approximates the partial monotone formula size or the partial monotone circuit size.
\end{theorem}

The formal statement of Theorem~\ref{thm:introthm2} is in the fourth item of Corollary~\ref{cor:minlt-runtime}. 

Both Theorem~\ref{thm:introthm1} and Theorem~\ref{thm:introthm2} are consequences of a
more general hardness result for a monotone version of the Computational Gap Learning
problem studied in~\cite{ABX08}. In our problem, which we call Monotone Circuit-Formula
Gap Learning, we are given a circuit which samples a distribution of labelled examples
$(x_i,b_i)$ where $x_i$ is of length $n$ and $b_i$ a bit, and are asked to distinguish
between 2 cases: the first where the distribution is completely consistent with a formula
of size at most $s_1$, and the second where the distribution does not even non-trivially
correlate with any circuit of size at most $s_2$. Here $s_1, s_2: \mathbb{N} \rightarrow
\mathbb{N}$ are parameters of the problem, which we are able to choose so that $s_2 \gg
s_1$. Theorem~\ref{thm:introthm1} follows immediately from our result about Monotone
Circuit-Formula Gap Learning, while Theorem~\ref{thm:introthm2} follows by a standard
sampling argument together with a union bound.

To show our central result about Monotone Circuit-Formula Gap Learning, we use a novel
technique in the domain of hardness of learning: {\it lifting}, together with known strong
inapproximability results for Resolution proof
size~\cite{DBLP:journals/jacm/AtseriasM20,DBLP:conf/stoc/RezendeGNPR021,DBLP:conf/stoc/GoosKMP20,DBLP:conf/lagos/Rezende21,
DBLP:conf/coco/CareniniR25}. One of the main applications
of lifting is to obtain monotone circuit size lower bounds in a generic way from
Resolution proof width lower bounds~\cite{Garg2020, GKRS19Adventures}. The generic nature of
these reductions opens the possibility of using them to derive {\it reductions} from
approximating proof size for various proof systems to approximating circuit size for
various circuit models. This is yet another instance of the connection between learning
and automatability of proof systems proposed in~\cite{ABFKP08}.

It may be useful to compare our results with the work
of Hirahara~\cite{H22}. Hirahara shows that a variant of the Computational Gap Learning problem
is $\NP$-hard (under randomised reductions),
where the task is to distinguish between linear programs (Turing machines
computing a linear function) of size $s_1$
and programs of size $s_2 = s_1 \cdot n^{o(1)}$.
Besides the obvious difference that we consider monotone circuits vs.\
formulas
whereas~\cite{H22} considers programs,
crucially our results allow for a superpolynomial gap between circuit size,
whereas~\cite{H22} obtain only a sublinear gap between program size (though
    under the weaker assumption of $\NP$-hardness).
Hirahara also shows $\NP$-hardness of approximating partial circuit size,
when one is given the entire truth table of a partial function (referred to
as $\MCSP^*$),
again only with a $n^{o(1)}$ approximation gap.
In \myref{thm:introthm2}, however, we consider a \emph{succinct}
version of this problem, in which we are only given inputs where the
function is defined; we also obtain a stronger hardness of approximation gap here.

We next give a technical overview of the ideas behind our proofs.

\subsection{Technical Overview}

\myparagraph{Overall approach}

The main insight of the paper is to convert the hardness of
\emph{automating}
Resolution~\cite{DBLP:journals/jacm/AtseriasM20}
into the hardness of 
approximating monotone circuit size,
using 
results from the literature (called \emph{lifting theorems})
that connect the complexity of Resolution refutations with the monotone
complexity of Boolean functions~\cite{Garg2020}.

To understand how this works, 
let us review how the seminal paper of Atserias and Müller~\cite{DBLP:journals/jacm/AtseriasM20}
proved that automating Resolution is $\NP$-hard.
Recall that an algorithm that automates Resolution must, given as
input an unsatisfiable CNF formula $\phi$, construct a Resolution refutation of
$\phi$ in time $\poly(\abs{\phi} + s)$, where $s$ is the size of the shortest
Resolution refutation of $\phi$.
Given a 3-CNF formula $F$ on $n$ variables as input, the proof
of~\cite{DBLP:journals/jacm/AtseriasM20} consists in
constructing in polynomial-time a formula $\REF(F)$
with the following properties:
\begin{enumerate}
    \item If $F$ is satisfiable, then $\REF(F)$ admits a polynomial-size
        Resolution refutation;
    \item If $F$ is unsatisfiable, then any Resolution refutation of
        $\REF(F)$ must have size $2^{n^{\Omega(1)}}$.
\end{enumerate}
In turn, this implies that if we can \emph{automate} Resolution,
then we can determine in which of the two cases we are in, and thus 
decide the satisfiability of $F$.

Another significant ingredient in our approach 
is the \emph{DAG-like lifting theorem},
which converts Resolution \emph{width} lower bounds 
for a formula $\phi$
into 
monotone circuit size lower bounds for a related monotone function $f_{\phi}$~\cite{Garg2020}.
Importantly, the monotone complexity of the 
resulting function $f_{\phi}$ is tightly connected to the Resolution width of $\phi$,
so that $f_{\phi}$ has large monotone complexity iff $\phi$ requires large Resolution
width.

Ideally, we would now like to combine the properties of $\REF(F)$ with the lifting theorem
to reduce the satisfiability of $F$ to deciding the monotone complexity of
$f := f_{\REF(F)}$
when
given access to some succinct representation of $f$
(e.g., access to random samples of $f$, or a succinct representation of its truth-table).
Unfortunately, the connection between monotone complexity and Resolution goes only via
Resolution \emph{width},
whereas 
the complexity gap obtained by Atserias-Müller is only for Resolution \emph{size}.
Though Resolution size does indeed give \emph{upper bounds} on the monotone complexity of
a related function via \emph{feasible
interpolation}~\cite{krajivcek1997interpolation,kraivcek1998interpolation,razborov1995unprovability},
the converse is not known to hold.

A potential way out would be to consider the notion of \emph{block-width} of
proofs~\cite{DBLP:journals/jacm/AtseriasM20}, 
a generalisation of Resolution-width to the setting when the variables are split
into blocks, as is the case with $\REF(F)$.
The proof of Atserias-Müller can be rephrased as first showing that the block-width of
$\REF(F)$ is $O(1)$ when $F$ is satisfiable,
and $n^{\Omega(1)}$ when $F$ is unsatisfiable~\cite{DBLP:conf/stoc/GoosKMP20}.
However, though the notion of block-width can be lifted to Resolution
size,
it's not known if it can be lifted to monotone circuit size,
and indeed significant technical challenges have been found
when trying to do so in at least one previous work~\cite{DBLP:conf/stoc/GoosKMP20}.

A more promising approach comes from considering a different encoding
of $\REF(F)$ found in~\cite{DBLP:conf/stoc/RezendeGNPR021}.
Their encoding admits a Resolution-width gap of $\REF(F)$ 
of the following form\footnote{The encoding of~\cite{DBLP:journals/jacm/AtseriasM20} cannot achieve this gap because it uses unary pointers.}:
\begin{enumerate}
    \item If $F$ is satisfiable, $\REF(F)$ has Resolution-width $O(n)$;
    \item If $F$ is unsatisfiable, $\REF(F)$ has Resolution-width $\Omega(n^3)$.
\end{enumerate}
This seems to be what we are aiming at: an efficiently constructible formula
with a noticeable gap in Resolution-width.
However,
lifting gives us that the monotone complexity of $f_{\phi}$
is roughly
$2^{\Theta(\wRes(\phi))} \cdot \Res(\phi)$.
So we obtain a function on $n^{O(1)}$ input bits
with $2^{O(n)}$ monotone complexity when $F$ is satisfiable,
and $2^{\Omega(n^3)}$ when $F$ is unsatisfiable.
Unfortunately,
it is not clear
how to distinguish between
these two cases other than by guessing the \emph{exponential size} circuit
that computes the function.
In other words, we have obtained a reduction from an $\NP$ problem
to an $\NEXP$ problem, which is hardly surprising.
Moreover, the reduction is ultimately useless, as we could just solve $\SAT$ directly
in exponential time.

We overcome this hurdle
by constructing a different formula $\REFx(F)$
with the following properties:\footnote{For simplicity of exposition, at this introduction
we only explain the construction with the parameters that give a polynomial gap between
linear-size formulas and monotone circuits. The construction for ``juntas'' (i.e.,
functions that do not depend on all inputs) is the same, but
with a different parametrisation.}
\begin{enumerate}
    \item $\REFx(F)$ has $2^{O(\sqrt{n})}$ variables and $2^{O(\sqrt{n})}$ clauses.
    \bnote{Maybe we focus on $\sqrt{n}$ here to make it easier for the reader?}
    \item $\REFx(F)$ admits Resolution proofs of depth at most $O(\sqrt{n})$ when $F$ is satisfiable;
    \item $\REFx(F)$ requires Resolution proofs of width at least $2^{\Omega(\sqrt{n})}$
    when $F$ is unsatisfiable.
\end{enumerate}
Note that we obtain not just a Resolution width upper bound in the satisfiable case,
but the strictly more powerful upper bound on Resolution \emph{depth}, which is strongly tied to the size of monotone
\emph{formulas}~\cite{raz1999separation, goos2015deterministic}.
This is how we obtain a gap between monotone formula and circuit sizes in the final
result.
Note moreover that the resulting formula has \emph{subexponential} size.
This is a careful tradeoff that we exploit:
by increasing the size of the formula,
we are able to reduce its Resolution-width (and even its Resolution-depth).
We now describe the construction of the formula in more detail.

\myparagraph{The $\REF^*(F)$ Formula}

The formula $\REF(F)$ essentially encodes the statement ``$F$ has a Resolution proof of
small size~$s$'', where $s$ is a fixed polynomial in the number of variables of $F$. 
When $F$ is satisfiable, the formula $\REF(F)$ is \emph{unsatisfiable}, 
and this formula can be refuted with low \emph{block-width}~\cite{DBLP:journals/jacm/AtseriasM20,Pudlak03}.

Their strategy can be seen as
low memory technique to achieve the following.
Given any proposed proof for $F$,
a \emph{refuter}\footnote{The ``Refuter'' is typically referred to as ``Prover'' in a Prover-Adversary game. Since $F$ is satisfiable, the goal of the Prover is to find a contradiction to the statement that $F$ admits a Resolution refutation, and so its task can be referred to as ``refuting'' a proposed proof of unsatisfiability of $F$.} navigates through it
while maintaining
the guarantee that the current clause is unsatisfied by the known satisfying assignment
$x^*$. Beginning at~$\bot$, eventually such a refuter will make it to a line of the proof that is a weakening
of an axiom. Since the line must be unsatisfied by $x^*$, it cannot be a weakening of any
axiom and a contradiction has been found. While this technique can be used to show an
$O(1)$ block-width upper bound, the depth of the proof strategy is $\Omega(s \cdot n)$, since it
will involve learning every variable of $\Omega(s)$ lines of the proposed proof. Adapting
this technique to achieve a $\sqrt n$ depth upper bound seems like a nearly hopeless
endeavour on first blush. Achieving depth $\sqrt{n}$ requires us to refute $\REF(F)$
while querying in total fewer variables than the number of variables of $F$. This
presents two barriers:
\begin{itemize}
    \item The refuter needs to get to a line of the proof that is claimed to be a weakening of an axiom while only seeing variables from $\sqrt n$ total blocks, and
    \item The refuter needs to learn some information about all $n$ variables while being able to query far fewer than $n$ variables.
\end{itemize}

To overcome the first of these issues, we can force the proof to have depth $\sqrt n$, as done in~\cite{DBLP:conf/lagos/Rezende21}.
However, if we are only resolving over one literal at a time, the clauses at depth-$\sqrt{n}$, which are weakening of axioms of $F$, would have width $\sqrt{n}$. This would break the
lower bound of Atserias and Müller because their strategy requires
that the clauses in the purported proof 
that are claimed to be a weakening of
an axiom must have width $n$.

To retrieve the lower bound, we can change the underlying proof system to one
which resolves on $\sqrt{n}$ variables at a time. This results in each block, instead of
having just a left and a right child---which are the two clauses from which it was derived---%
now having $2^{\sqrt{n}}$ children,
and thus having one pointer $\point_z$
for each possible setting $z \in \bit^{\sqrt n}$ of the variables being
resolved over (we also importantly fix the order in which variables are resolved over,
with the first segment of $\sqrt n$ variables, then the second segment, etc).
Unfortunately, this immediately means that the resulting $\REF(F)$ formula will have to be of size at
least $2^{\sqrt n}$ just to manage these pointer variables. On the positive side, this
gives us a hint for how to overcome the second barrier.

In standard resolution refutations, if you know that a clause $B'$ is the left child of $B$ which
is the result of resolving over $x_i$, you know that $x_i$ is in $B'$, and more importantly that $\lnot x_i$ is \emph{not} in $B'$; and if $B'$ is the
right child you know that $\lnot x_i$ is in $B'$, and more importantly that $x_i$ is \emph{not} in $B'$. Consequently,  knowing that $x^*$ is a satisfying
assignment, if the prover always queries whichever child is not satisfied by $x^*$, they will
eventually find themselves in a leaf which cannot be a valid weakening. In our new multi-variable resolution system, if you know that $B'$ is the $z$\textsuperscript{th} child of
$B$, you know $\sqrt{n}$ bits of information about the clause that $B'$ claims to have. 
Using
$x^*$ as above, one could in principle get to a leaf which cannot be a valid weakening in
depth $O(\sqrt n)$.

We can extend this density of information from the pointers to a new
class of summary variables $\mathsf{summ}^B_{i,z}$ which give us information
about all the variables from $x_{(i-1)\sqrt{n}}$ to $x_{i\sqrt{n}-1}$. 
This idea of using summary variables
appeared before in~\cite{DBLP:conf/coco/CareniniR25}.
It is important that
our summary variables encode which literals are
\emph{not} present (which is ultimately what we are concerned with for the upper bound). Formally,
$\summlit_{i,z}^{B}$ is set equal to $1$ to indicate that for each $z_j = 1$, literal
$x_{(i-1)\formuladepth + j}$ is not present in block $B$, and that
for each $z_j = 0$, literal $\lnot x_{(i-1)\formuladepth + j}$ is not present in block $B$.
We can ensure that
a Prover can learn what all $\sqrt{n}$ of
the $\summlit_{i,z}^{B}$ variables for a leaf node should be while only querying the right
set of $O(\sqrt{n})$ variables on their way from the root to a leaf. %

To make our bounds as tight as possible, we add some additional complexity, including
binary pointers for axioms, split variables to reduce clause width, and a low degree
expander between each layer of the proof. This low-degree expander,
which was previously used in~\cite{DBLP:conf/lagos/Rezende21,DBLP:conf/coco/CareniniR25},
allows us to keep down
to a constant the
number of $\point_{B',z}^B$ variables a prover needs to query for each $B,z$ before
finding one that is set to one. While the expander does somewhat complicate the
proof of the lower bound, we are able to derive it by reducing the retraction
pigeonhole principle (rPHP) instance to finding a violated clause, similarly to what was done in~\cite{DBLP:conf/lagos/Rezende21}.

\myparagraph{Hardness of approximating circuit size and learning}

To see how to join our new $\REF^*$ formula with a lifting theorem
to obtain the claimed hardness results, it will be helpful to understand the function that
comes from lifting.
Roughly speaking, ``lifting'' usually consists in composing some original
computational object (e.g., a propositional formula or a function)
with a gadget, so as to hide the original object in a stronger computational model
in such a way that the only way to compute/solve the composed problem is by solving the
original problem.

In our case, given a formula $F$ with $n$ variables 
and width $k$,
together with a parameter~$m$,
the lifted function $f := f_{F,m}$ from~\cite{Garg2020}
will have $\card{F} \cdot m^k$ input bits,
where $\card{F}$ is the number of clauses of $F$.
To each $i \in [\card{F} \cdot m^k]$,
we associate a pair $(C, \alpha)$,
where $C$ is a clause of $F$
and $\alpha \in [m]^k$.
We interpret
$(C, \alpha)$
as the $k$-width clause over the variables
$\set{y_{ij} : i \in [n], j \in [m]}$
obtained by
replacing the variable $x_i \in \vars(C)$
in $C$ with $y_{i, \alpha(i)}$.
We can thus see each input to $f$
as a collection of clauses $(C,\alpha)$.
The task of the partial function $f_{F,m}$ is to distinguish between the following two cases:
\begin{enumerate}
    \item Accept (Output 1): there exists $x : [n] \to [m]$ such that the input bit corresponding to
        $(C,x^{\vars(C)})$ is set to 1 for every clause $C$ of $F$, and everything else is
        set to 0;
    \item Reject (Output 0): there exists $y \in \blt^{mn}$ such that all input bits
        associated with clauses $(C,\alpha)$ satisfied by
        $y$ are set to $1$, and everything else is set to 0.
\end{enumerate}
The partial function is clearly well-defined, as in the first case the collection of
clauses is unsatisfiable, and in the second case they are satisfiable.
For this same reason, no input to the function which contains all the clauses of Case (1)
will ever be in Case (2), so the function is monotone.
One can think of the function $f$ as ``hiding'' several variable renamings
of the the formula $F$ 
in its input bits; to compute $f$ is to ``sniff out'' if $F$ has been completely embedded
among the renamed clauses.

Most importantly for our purposes, 
we can easily sample from the distribution which half of the time
gives a uniform accepting input of $f$,
and in the other half gives a uniform rejecting input.
Indeed, first sample a bit $b \in \blt$,
then, if $b = 1$, sample $n \log m$ bits and interpret it as $x : [n] \to [m]$,
otherwise sample a string $y \in \blt^{mn}$. With $x$ (resp. $y$), it's then easy
to set the relevant bits so as to output an accepting (resp. rejecting) input of
$f_{F,m}$.
Let us abuse notation and denote by $\cald^{F,m}$ both a circuit that samples
from this distribution of binary strings (together with the bit $b$ as a prefix) 
and the distribution itself.
We thus see samples from $\cald^{F,m}$ as a pair $(x,b)$ where $x \in
\blt^{\card{F} m^k}$ and $b \in \blt$.
Given $F$ and $m$ as input, it's also easy to see that we
can construct the circuit sampling from $\cald^{F,m}$ in time $\poly(\card{F}, m^k)$. 

Going back to the approach above, given an input $F$ to 3SAT,
our reduction will first construct
$\cald := \cald^{\REF^*(F), m}$,
so as to obtain a distribution with the following properties.
If $F$ is satisfiable, then
\begin{equation}
    \label{eq:sat-case}
    \text{
        there is a monotone formula $L$ of size
        $\leq s_1$
        s.t.
        $L(x)=b$ for every $(x,b) \in \supp(\cald)$;
    }
\end{equation}
if $F$ is unsatisfiable,
then
\begin{equation}
    \label{eq:unsat-case}
    \text{
        every monotone circuit $K$ of size $\leq s_2$ 
        satisfies
        $\Pr_{(x,b) \flws \cald}[K(x) = b] < 1/2+\gamma$.
    }
\end{equation}
The value of $s_1$ is 
$2^{O(\sqrt{n} \log m)}$,
and the value of
$s_2$ is 
$2^{\Omega(m / \sqrt{n})}$;
these come from a
refined version of the DAG-like lifting theorem~\cite{DBLP:conf/coco/RezendeV25}
with ``small gadgets''
that allows us 
to essentially pick any value of $m$ we would like.
However, lifting theorems typically are only stated for $\gamma = 1/2$
(corresponding to a worst-case lower bound).
Fortunately, generalising to $\gamma \approx m^{-1}$ is an easy task,
requiring only the modification of certain numbers and parameters in the proof
of~\cite{DBLP:conf/coco/RezendeV25}.
This is essential for our purposes, as in a learning problem
we only expect to approximate the target concept.

We are now ready to state the computational problem for which we show a lower bound.
The \emph{Monotone Circuit-Formula Gap Learning problem} $\mCFGL_{s_1}^{s_2}[\gamma]$
is the problem of distinguishing between cases 
$(\ref{eq:sat-case})$
and
$(\ref{eq:unsat-case})$
when given $\cald$ as input.
Given that our $\REFx(F)$ formula has size $2^{O(\sqrt{n})}$,
the reduction from SAT to this problem takes this time.
Under ETH, we thus obtain an 
$N^{\Omega(\log N)}$ 
lower bound for this problem, since the size $N$ of the input is
$2^{O(\sqrt{n})}$.
To reduce to the case $s_1 = n$, we pad the sampled strings;
this takes time $2^{\tld{O}(\sqrt{n})}$,
so our final lower bound is
$N^{\tld{\Omega}(\log N)}$.

We can further reduce this problem to a gap version of Partial-Monotone-MCSP
we call
$\mMCFSP_{s_1}^{s_2}[\gamma]$\footnote{Sometimes the literature reserves
the term ``Partial-MCSP'' to the problem where one is given the entire
truth table of a partial function, but we follow here the convention
of~\cite{ILO20}.}.
In this problem, instead of $\cald$, the algorithm is given
$t$ examples $(x_1,b_1), \dots, (x_t,b_t)$,
and must distinguish whether there is a monotone formula $L$ of size $s_1$
such that $L(x_i) = b_i$ for all $i \in [t]$,
or whether every monotone circuit $C$
of size $s_2$ is such that 
$\card{\set{i \in [t] : C(x_i) \neq b_i}} \geq (1/2+\gamma)t$.
The reduction is randomised, naturally coming from sampling the distribution~$\cald$.
Moreover,
we need to sample $O(s_2 \log s_2)$ 
examples
so that, with high probability, any circuit of size at most $s_2$
will err on a proportion larger than $1/2+\gamma$ of the examples
(in case~(\ref{eq:unsat-case})), and so the reduction
takes time $O(s_2 \log s_2)$.
This means that we get a tradeoff between the choice of $s_2$ and the runtime lower
bound coming from rETH.
The same reduction also gives us lower bounds on PAC-learning.

\section{Preliminaries}

For $f,g : \bbn \to \bbn$,
we write $f \sim g$ if $\lim f/g \to 1$,
and $f \asymp g$ if
$f = \Theta(g)$.
We say that a function Boolean $f$ is $\gamma$-approximated by $g$
over a distribution $\cald$
if $\Pr_{x \flws \cald}[f(x)=g(x)] \geq 1/2+\gamma$.

\subsection{Basic Complexity}

\emph{Monotone circuits} are Boolean circuits without $\neg$ gates.
\emph{Formulas} are Boolean circuits 
whose gates have fan-out 1.
\emph{Monotone formulas} are analogously defined.

The \emph{Exponential-Time Hypothesis (ETH)} postulates that there
exists a constant $\eps > 0$
such that
any algorithm solving 3SAT on $n$-variate formulas 
must take time $2^{\eps n}$.
The \emph{randomised} ETH \emph{(rETH)}
postulates the same lower bound for randomised algorithms.

\subsection{Expander Graphs}

A $(m,n,\delta,r,e)$-\emph{bipartite expander graph} is a bipartite graph $\graph=((U\cup V), E)$
with $|{U}| = m$, $|{V}| = n$, with every vertex in $U$ having degree at most $\delta$, and 
where for every subset $S\subseteq [m]$ such that $|S|\leq r$, the neighbourhood of $S$, denoted
$N(S)$, satisfies $|{N}(S)|\geq e|S|$. 
\begin{remark}
    \label{rk:expander}
    A complete bipartite graph $\graph=((U\cup V), E)$
    with $|{U}| = m$, $|{V}| = n$, is a 
    $(m,n,\delta,r,c)$-{bipartite expander graph} for, e.g., $\delta=n$, $r=n/2$ and $c=2$.
\end{remark}

A standard calculation~\cite{hoory06}
shows that random graphs are
good expanders. We record one parameter regime that we use in our applications.

\begin{lemma}
  \label{lem:random-expander-constant}
    Let $m,n \in \mathbb{N}$ be such that $m = n^2$. Then there exists $\delta = O(1)$
and $r = \Omega(\sqrt{n})$ such that
with high probability %
a random graph $\graph \sim \mathbf{G}(m,n,\delta)$ is an $(m,n,\delta,r,\delta/2)$-bipartite expander graph.
\end{lemma}

\subsection{Proof Complexity}
\label{sec:proof-complexity}

\myparagraph{Resolution}
A \emph{literal} is a propositional atom or its negation, a \emph{clause}
is a disjunction of literals, and a \emph{CNF formula} is a conjunction of
clauses. 
We see clauses as sets of literals, and write simply $C \subseteq D$ to express that $C$ is a subclause of $D$.  
We denote the set of variables in a clause $C$ by $\vars(C)$ 
(if $x_i$ or $\neg x_i$ are
present in $C$, we have $x_i \in \vars(C)$).
A \emph{Resolution
refutation} of an unsatisfiable CNF formula $\varphi = C_1 \land \dots
\land C_m$ over variables~$x_1, \dots, x_n$ is a sequence $D_1, \dots, D_s$
of clauses over $x_1, \dots, x_n$ such that $D_s = \bot$, denoting the
empty clause, and for every $i \in [s-1]$, the clause $D_i$ either (a) is
one of the clauses $C_1, \dots, C_m$ of $\varphi$, or (b) is a
\emph{weakening} of a previous clause, meaning that $D_i \supseteq D_j$ for
some $1\leq j < i$, or (c) has been obtained from two previous clauses $D_j
= A \lor x$ and $D_k = B \lor \neg x$, for~$j, k < i$, by an application of
the \emph{Resolution rule}:
    \begin{prooftree}
        \AxiomC{$A \lor x$}
        \AxiomC{$B \lor \neg x$}
        \BinaryInfC{$A \lor B$}
    \end{prooftree}
We say that $A \lor B$ was obtained by \emph{resolving} over $x$.
The \emph{size} of the refutation is~$s$, the number of clauses appearing in the refutation.
 
To every Resolution refutation 
we can associate a directed acyclic
graph (dag) in a natural way. We 
define
the \emph{depth} of the proof as
the length of the longest path in the dag, starting from the
root labeled by the empty clause $\bot$.
The proof is said to be \emph{tree-like} if the associated graph
is a tree rooted at $\bot$.

Let $F$ be an unsatisfiable $k$-CNF formula.
We denote by $\wRes(F)$ the minimum integer such that
$F$ admits a Resolution refutation whose clauses all have width at
most $\wRes(F)$.
Let $\ResDepth(F)$
be the minimum integer such that $F$ admits a Resolution refutation
with depth at most $\ResDepth(F)$.
We denote by $\Res(F)$
the size of the smallest Resolution refutation of $F$,
and by
$\TreeRes(F)$ 
the size of the smallest tree-like such refutation.
We denote by $\card{F}$
the number of clauses in $F$.

\myparagraph{Prover-adversary game}

An equivalent characterisation of Resolution width and Resolution depth can
be given by notions from \emph{query complexity}.
Suppose that $F$ is an $n$-variate unsatisfiable CNF formula.
Given any assignment $x \in \blt^n$ to the variables of $F$,
some clause of $C$ is necessarily falsified.
The minimum number $k$ of queries to the variables of $F$ needed to discover one such clause
is equivalent to $\ResDepth(F)$.

To characterise $\wRes$, it is helpful to modify
the above setup in what is commonly called 
a ``Prover-Adversary game''. As before, the Prover is querying
variables,
but now the Prover is allowed to forget the value of some queried
variables.
The adversary is answering the queries, but may respond with different
values
when a variable whose value was forgotten is queried again.
The minimum number $k$ of variables that the Prover needs to store
in order to discover a falsified clause is equal to $\wRes(F) \pm
1$~\cite{DBLP:conf/coco/AtseriasD03}.

\subsection{Lifting for Average-Case Lower Bounds}

Given a \emph{gadget}
$g : \calx \times \caly$, let
$g^n : \calx^n \times \caly^n$
denote the function
\begin{equation*}
    g^n(x,y) = (g(x_1,y_1),\dots,g(x_n,y_n)).
\end{equation*}
Let $\Ind_{m} : [m] \times \blt^m \to \blt$ 
denote the ``index gadget'',
defined as
the Boolean function such that $\Ind_m(x,y) = y[x]$.

\begin{definition}
    \label{def:function-lifting}
    Let $F$ be an unsatisfiable formula with $n$ variables.
    Let $f_{F,m}$ be the partial Boolean function over
    $\card{F} \cdot m^k$ bits
    defined as follows.
    To each $i \in [\card{F} \cdot m^k]$,
    associate a pair $(C, \alpha)$,
    where $C$ is a clause of $F$
    and $\alpha \in [m]^k$.
    We interpret
    $(C, \alpha)$
    as the $k$-width clause over the variables
    $\set{y_{ij} : i \in [n], j \in [m]}$
    obtained by
    replacing the variable $x_i \in \vars(C)$
    in $C$ with $y_{i, \alpha(i)}$.
    Each binary string of length $\card{F} \cdot m^k$
    is thus in 1-to-1 correspondence
    with the set 
    \begin{equation*}
        F \circ \Ind_m^n
        :=
        \set{(C,\alpha) : \text{$C$ is a clause of $F$},\;
            \alpha : \vars(C) \to
        [m]}.
    \end{equation*}
    The function $f_{F,m}$ has the following behaviour:
    \begin{enumerate}
        \item 
            Accept the
            binary string whose support is
            $\set{(C,\alpha) \in F \circ \Ind_m^n : x^{\vars(C)} = \alpha}$,
            for every
            $x : [n] \to [m]$.
        \item 
            Reject the binary string whose support is
            $\set{(C,\alpha) \in F \circ \Ind_m^n : (C,\alpha)(y)=1}$,
            for every $y \in \blt^{nm}$.
    \end{enumerate}
\end{definition}

The following theorem shows that the monotone circuit complexity of
$f_{F,m}$
is roughly $m^{\Theta(\wRes(F))}$
for some $m \gg \wRes(F)$; moreover, a formula upper bound of
the form
$m^{O(\ResDepth(F))}$ also holds.
This result follows from (a small modification of) the so-called
\emph{lifting theorem},
which first appeared in~\cite{Garg2020},
from which monotone circuit lower bounds 
follow from Resolution width lower bounds.
We use the lifting theorem from~\cite{DBLP:conf/coco/RezendeV25}, which
builds on~\cite{dRFJNP24, Lovett2022}, as it not only has better parameters, but also gives a lower bound for $f_{F,m}$ for smaller values of $m$, which is essential to our purposes.

We say that a circuit $D : \blt^r \to \blt^N$
\emph{represents} a distribution $\cald$ over
$N$-input bits if 
the random binary string $D(\rndr)$ where $\rndr \flws \blt^r$
is distributed according to~$\cald$.

\begin{theorem}[\protect{\cite[Theorem 11, modified]{DBLP:conf/coco/RezendeV25}}]
    \label{thm:constructive-lifting}
    There exists an algorithm $A(F, 1^m)$ 
    which receives as input
    an unsatisfiable $k$-CNF formula $F$ on $N$ variables,
    runs in time
    $\poly(\card{F} \cdot m^k)$,
    and outputs a circuit
    representing a distribution
    $\cald^{F,m} = (\cald_1 + \cald_0)/2$ 
    over $\card{F} \cdot m^k$
    bits 
    such that:
    \begin{enumerate}
        \item 
            There is a monotone circuit $K$ of size
            $m^{O(\wRes(F))} \cdot \Res(F)$
            and a monotone formula~$L$ of size
            $m^{O(\ResDepth(F))}$
            such that
            $K(x) = L(x) = b$
            for every $x \in \supp(\cald_b)$ and $b \in \blt$.
        \item
            There exists an absolute constant $c$ such that for any $\ell \leq \wRes(F)$, any monotone circuit $K$ that satisfies
            $\Pr_{b \flws \blt, x \flws \cald_b}[K(x)=b] \geq 1/2 + \gamma$, for $\gamma \geq c \cdot \ell \cdot \log(mN) / m$,
            has size at least
            \begin{equation*}
                \left( 
                    \frac{m}{c \cdot \ell \cdot \log(mN)}
                \right)^{\ell -1}.
            \end{equation*}
    \end{enumerate}
    Moreover, the distribution $\cald_0$
    (resp. $\cald_1$)
    is uniformly distributed over $f_{F,m}^{-1}(0)$
    (resp. $f_{F,m}^{-1}(1)$).
\end{theorem}

\begin{proof}[Proof sketch]
    The algorithm $A$
    constructs a circuit $\cald^{F,m} : \blt^r \to \blt^{\card{F} \cdot m^k}$
    where $r = mN+1$ that behaves in the following way.
    If the first bit of $r$ is 0,
    then 
    let
    $y \in \blt^{mN}$ be the remaining $mN$ bits.
    In the output bit $(C,\alpha)$, the circuit outputs $1$
    iff
    $y$ satisfies
    the clause
    $(C,\alpha)$.
    Moreover, if the first bit of $r$ is 1,
    then
    let
    $x \in [m]^N$
    be given the first $m \log n$ of the remaining bits.
    In the output bit $(C,\alpha)$, the circuit outputs $1$
    if
    $x^{\vars(C)} = \alpha$.
    The construction of $D$, and $D$ itself,
    can be clearly be done in time
    $\poly(\card{F} \cdot m^k)$.

    Property (1) of the resulting distribution is folklore
    (see, e.g.,~\cite{Lovett2022}).
    Property (2) is proved 
    in~\cite[Theorem 11]{DBLP:conf/coco/RezendeV25} with $\gamma = 1/2$;
    the result with smaller $\gamma$ 
    can be easily obtained by tweaking a few parameters. 
    We explain how this is done in
\ifthenelse{\conferenceversion=1}{%
the full version of this paper.}{%
\myref{sec:refined-lifting}.}
\end{proof}

\section{Shallow Refutation Formula}

Let $F$ be a 3-CNF formula over $n$ variables.
The $\REF_s(F)$ formula encodes the statement
``$F$ has a resolution proof of length $s$''.
Slightly more formally, the variables of $\REF_s(F)$
are divided into $s$ blocks, each of which encodes
a potential clause in the purported refutation of $F$ and how this clause was derived,
and the clauses of 
$\REF_s(F)$ enforce that all derivation steps are sound.

The known resolution refutation of $\REF_s(F)$ when $F$ is satisfiable, given by Pudlák's upper 
bound~\cite{Pudlak03}, has large width $\Omega(n)$. Under ETH, a large width is unavoidable, 
since it is possible to find width $w$ refutations of an unsatisfiable $N$-variate CNF formula in time 
$N^{O(w)}$.
While it is possible to show a gap between the Resolution width of the formula $\REF_s(F)$
in the satisfiable and unsatisfiable cases of~$F$~\cite{DBLP:journals/jacm/AtseriasM20,DBLP:conf/stoc/RezendeGNPR021} ($O(n)$ vs.\ $\Omega(s/n)$),
for our applications we need a better upper bound.
In this section, we define a formula related to $\REF_s(F)$ which substantially decrease the width upper bound at the cost of an increase in the size of the formula. Moreover, we show that our new formula has small Resolution \emph{depth} when $F$ is satisfiable.

\subsection{The New Formula \texorpdfstring{$\eREF_d(F)$}{Ref(F)}}
Our new formula is denoted by $\eREF_d(F)$, and it depends on a bipartite expander graph~$G$ and an integer parameter $d$.
We somewhat follow the notation established in ``The Proof Analysis
Problem'' by Arteche et al. \cite{Arteche2025}, but modify the structure of the purported resolution refutation and introduce some extension variables,
similar to what was done in
\cite{DBLP:conf/lagos/Rezende21} and
\cite{DBLP:conf/coco/CareniniR25}.

As was done in~\cite{DBLP:conf/lagos/Rezende21} we structure the blocks in $\formuladepth$ layers (instead of $n$), and as in~\cite{DBLP:conf/coco/CareniniR25},
we add 
``summary'' variables that give information about segments of $\noverdepth$
variables.
These changes allow us to turn the $O(1)$ block-width upper bound and
corresponding $O(n)$ resolution width upper bound in~\cite{Pudlak03} into an $O(\formuladepth)$ upper bound
on depth, at the cost of increasing the number of variables and clauses in
$\eREF_d(F)$ to $2^{\Theta(\noverdepth)}$.

Since it is crucial that these summary variables encode that literals are absent,
we also switch focus and define ``un-literal'' variables encoding that single literals are \emph{not present} in a block as opposed to previous encoding of $\REF$ %
which
has variables encoding the \emph{inclusions} of literals. 
This is done only for exposition purposes since our $\lit$ variables are just the negation of the usual literal variables. 
We also include binary variable for pointers (as done in~\cite{DBLP:conf/stoc/RezendeGNPR021}) to make it possible to query pointers more efficiently.

Finally, we also include ``split'' variables to turn our formula into a $3$-CNF formula using 
a standard transformation.
This is done in order to minimise the blow-up in size of the function we obtain after lifting $\eREF_d(F)$.

We start by explaining the structure of the formula $\eREF_d(F)$. We then introduce the main 
variables of $\eREF_d(F)$ 
and the clauses enforcing that the purported proof is sound and is structured as we would like.  
Finally, we
introduce the split variables and explain how they are used to transform the formula into a $3$-CNF 
formula. 

\myparagraph{Structure of $\eREF_d(F)$ and notation} 
Let $F$ be a CNF formula over $n$ variables and $\poly(n)$
clauses, let $d$ be an integer, and let $G$ be a bipartite graph\bnote{specify params of expander graph here}
with $t2^{n/d}$ nodes on left side, $t$ nodes on the right, and outdegree $\delta$ from each node on the left.
We define $s \coloneqq t \cdot d$.
For the sake of convenience in our presentation,
we assume that $d$ divides~$n$,
and that $|F|$ (the number of clauses of $F$) and $\delta$, the left outdegree of the expander $G$,
are powers of 2,
so they can be cleanly indexed into by binary pointers. 
We later explain why we can make these assumptions without loss of generality~(see \myref{rk:divisibility}).

Given an integer $S$ which is a power of $2$ 
and an integer $X\in[S]$, we denote by $\bin_S(X)$ the binary string of length $\log S$ encoding $X$. When $S$ is clear from the context, we simply write $\bin(X)$.
For instance, for $A \in [|F|]$, 
we write $\bin(A) \in \bit^{\log |F|}$ to be 
the binary string of length $\log |F|$ encoding $A$.

The formula $\eREF_d(F)$ is defined over $s+1$ blocks arranged in $d+1$ layers (see \myref{fig:structure}).
The block $s+1$ is the root and it is alone in layer $0$. 
The other $s$ blocks are arranged in decreasing order in layers $1$ to $d$ with $s/d$ blocks per 
layer. 
Formally, for $B \in [s]$, we have $\layer(B) = 1+d - \lceil B/t \rceil$ while $\col(B) = B - (\lceil B/t \rceil - 1)t$. So, for example, for block~$s$ we get that $\layer(s) = 1, \col(s) = t$, while for block~$1$ we get $\layer(1) = d, \col(1) = 1$.

\begin{figure}[t]
\centering
\begin{tikzpicture}[
    font=\Large,
    box/.style={minimum width=3.8cm, minimum height=1cm, rounded corners=12pt},
    plainbox/.style={box,draw,thick},
      collabel/.style={anchor=center},
      scale = 0.8
]

\newlength{\vsep}
\newlength{\hsep}
\newlength{\bigvsep}
\newlength{\bighsep}

\setlength{\hsep}{1cm}
\setlength{\vsep}{1cm}
\setlength{\bighsep}{3cm}
\setlength{\bigvsep}{2cm}

\node[plainbox] (a1) {$(d-1)t + 1$};
\node[plainbox, right=\hsep of a1] (a2) {$(d-1)t + 2$};
\node[plainbox, right=\hsep of a2] (a3) {$(d-1)t + 3$};
\node[plainbox, right=\bighsep of a3] (a4) {$dt = s$};
\node at ($(a3)!0.5!(a4) + (0,0)$) {$\scalebox{2}{$\ldots$}$};

\node[plainbox,above=\vsep of a1] (bott) {$s+1$};

\node[plainbox, below=\vsep of a1] (b1) {$(d-2)t + 1$};
\node[plainbox, below=\vsep of a2] (b2) {$(d-2)t + 2$};
\node[plainbox, below=\vsep of a3] (b3) {$(d-2)t + 3$};
\node[plainbox, below=\vsep of a4] (b4) {$(d-1)t$};
\node at ($(b3)!0.5!(b4) + (0,0)$) {\scalebox{2}{$\ldots$}};

\node[plainbox, below=\vsep of b1] (c1) {$(d-3)t + 1$};
\node[plainbox, below=\vsep of b2] (c2) {$(d-3)t + 2$};
\node[plainbox, below=\vsep of b3] (c3) {$(d-3)t + 3$};
\node[plainbox, below=\vsep of b4] (c4) {$(d-2)t$};
\node at ($(c3)!0.5!(c4) + (0,0)$) {\scalebox{2}{$\ldots$}};

\node[plainbox, below=\bigvsep of c1] (d1) {$2t + 1$};
\node[plainbox, right=\hsep of d1] (d2) {$2t + 2$};
\node[plainbox, right=\hsep of d2] (d3) {$2t + 3$};
\node[plainbox, right=\bighsep of d3] (d4) {$3t$};
\node at ($(d3)!0.5!(d4) + (0,0)$) {\scalebox{2}{$\ldots$}};

\node[plainbox, below=\vsep of d1] (e1) {$t + 1$};
\node[plainbox, below=\vsep of d2] (e2) {$t + 2$};
\node[plainbox, below=\vsep of d3] (e3) {$t + 3$}; 
\node[plainbox, below=\vsep of d4] (e4) {$2t$};
\node at ($(e3)!0.5!(e4) + (0,0)$) {\scalebox{2}{$\ldots$}};

\node[plainbox, below=\vsep of e1] (f1) {1};
 \node[plainbox, below=\vsep of e2] (f2) {2};
\node[plainbox, below=\vsep of e3] (f3) {3};
\node[plainbox, below=\vsep of e4] (f4) {$t$};
\node at ($(f3)!0.5!(f4) + (0,0)$) {\scalebox{2}{$\ldots$}};

\node at ($(c1)!0.5!(d1)+ (0,0.2)$) {\scalebox{2}{$\vdots$}};
\node at ($(c2)!0.5!(d2)+ (0,0.2)$) {\scalebox{2}{$\vdots$}};
\node at ($(c3)!0.5!(d3)+ (0,0.2)$) {\scalebox{2}{$\vdots$}};
\node at ($(c4)!0.5!(d4)+ (0,0.2)$) {\scalebox{2}{$\vdots$}};

\draw[decoration={brace,raise=50pt},thick,decorate]
      ([xshift=-15pt]bott.south) -- node[left=55pt] {$\layer =0$} ([xshift=-15pt]bott.north);

\draw[decoration={brace,raise=50pt},thick,decorate]
      ([xshift=-15pt]a1.south) -- node[left=55pt] {$\layer=1$} ([xshift=-15pt]a1.north);

\draw[decoration={brace,raise=50pt},thick,decorate]
      ([xshift=-15pt]b1.south) -- node[left=55pt] {$\layer=2$} ([xshift=-15pt]b1.north);
      
\draw[decoration={brace,raise=50pt},thick,decorate]
      ([xshift=-15pt]c1.south) -- node[left=55pt] {$\layer=3$} ([xshift=-15pt]c1.north);

\draw[decoration={brace,raise=50pt},thick,decorate]
      ([xshift=-15pt]d1.south) -- node[left=55pt] {$\layer=d-2$} ([xshift=-15pt]d1.north);

\draw[decoration={brace,raise=50pt},thick,decorate]
      ([xshift=-15pt]e1.south) -- node[left=55pt] {$\layer=d-1$} ([xshift=-15pt]e1.north);

\draw[decoration={brace,raise=50pt},thick,decorate]
      ([xshift=-15pt]f1.south) -- node[left=55pt] {$\layer=d$} ([xshift=-15pt]f1.north);
      
\draw[decoration={brace,raise=25pt},thick,decorate]
    ([xshift=-2pt]f1.east) -- node[below=35pt] {$\col = 1$} ([xshift=2pt]f1.west);

\draw[decoration={brace,raise=25pt},thick,decorate]
    ([xshift=-2pt]f2.east) -- node[below=35pt] {$\col = 2$} ([xshift=2pt]f2.west);

\draw[decoration={brace,raise=25pt},thick,decorate]
    ([xshift=-2pt]f3.east) -- node[below=35pt] {$\col = 3$} ([xshift=2pt]f3.west);

\draw[decoration={brace,raise=25pt},thick,decorate]
    ([xshift=-2pt]f4.east) -- node[below=35pt] {$\col = t$} ([xshift=2pt]f4.west);

\draw[<->, thick, black] 
($(a4.north east)+(1,0)$) --
node[right] {{$d$}}
($(f4.south east)+(1,0)$);

\draw[<->, thick, black] 
($(f1.south west)+(0,-2.5)$) --
node[below] {{$t = s/d$}}
($(f4.south east)+(0,-2.5)$);

\end{tikzpicture}
    \caption{The layout of blocks within the purported refutation, demonstrating the columns and layers.}\label{fig:structure}
\end{figure}
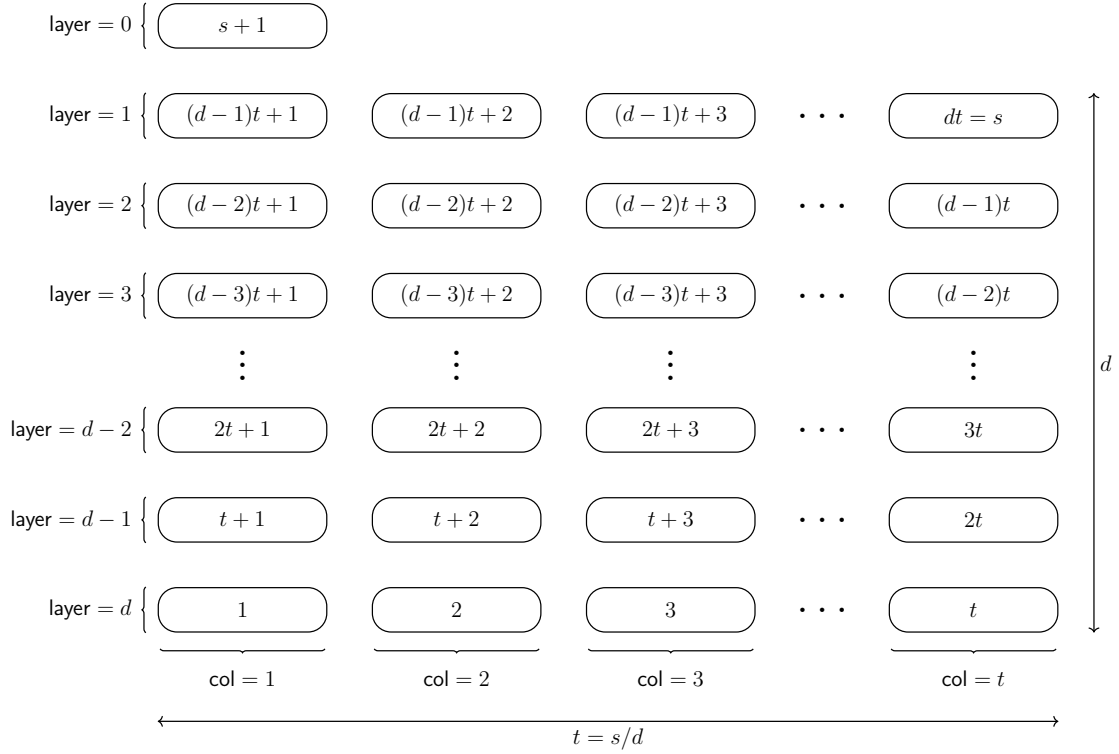

The idea is that each block contains a clause: block $s+1$ contains the empty clause, and blocks at layer $i$ contain clauses over the first $i\cdot \noverdepth$ variables.
For this to locally look like a sound refutation, each block is supposed to have $2^{\noverdepth}$ children on the layer below, one for each possible extension of the clause in the block to the next $\noverdepth$ variables: for each $z\in \{0,1\}^{\noverdepth}$, the clause in the $z$\textsuperscript{th} child of a block $B$ should
not be satisfied by assigning~$z$ these $\noverdepth$ variables.
The $z$\textsuperscript{th} child of a block is encoded by some variables so that, even though 
it is not possible to assign $2^{\noverdepth}$ distinct children to each block, 
locally it can still look like the formula $\REF^G_d(F)$ is consistent.

Finally, we would like there to be only a few possible 
options for each of the $z$\textsuperscript{th} children of a block. We use the graph $G$ to determine these options: each node on the left-side of~$G$ corresponds to a $z$\textsuperscript{th} child of a block at layer $i\in[d]$ and each node on the right-side of~$G$ corresponds to a block on layer $i+1$. The neighbour relation corresponds to the possible blocks at layer $i+1$ that could be the $z$\textsuperscript{th} child of a block at layer $i$. 
To make our encoding more succinct, we need a way of managing binary indices for the at most $\delta$ blocks $B'$ that could be a $z$\textsuperscript{th} child of $B$. To do this, we need a surjection (which we denote $\sur_{z,\col(B)}$) from the set of at most $\delta$ valid blocks that can be the $z$\textsuperscript{th} child to $[\delta]$. It is easy to see that, if $\sur_{z,\col(B)}$ is the lexicographical order over the valid blocks, then $\sur_{z,\col(B)}$ can be efficiently computed when needed while constructing $\REF^G_d(F)$ in all the parameter regimes considered in this paper. 
Note that there may be parameter regimes in which this is not the case.

\myparagraph{The variables of $\eREF_d(F)$} 
Let $x_1,..,x_n$ be the variables of $F$. For any variable $x_i$, the variable $(x_i)^0$ corresponds to $\lnot x_i$ and $(x_i)^1$ corresponds to $x_i$.
We define the set $\Lit_n = \{x_1,...,x_n,\overbar{x}_1,...,\overbar{x}_n\}$. 
Let $\textsf{Leaves} \coloneqq [s/d]$ be the set of blocks at layer $d$. 
Some variables (like $\weak, \weak\bin$, which encode what axiom in $F$ the block is a weakening 
of, in unary and in binary, respectively) are only defined for blocks in $\textsf{Leaves}$ and
some ($\point,\point\bin$, which encode the blocks that were used to derive the current block, also 
in unary and in binary, respectively) are 
only defined for blocks that are not in $\textsf{Leaves}$.

Recall that, in our setting,  a block is obtained by resolving $n/d$ variables at once.
Therefore, the $z$\textsuperscript{th} child, for $z\in \bit^{n/d}$, of a block at layer $i$ 
should not contain literals $x_{i\formuladepth + j}$ for each $z_j = 1$, 
nor literals $\lnot x_{i\formuladepth + j}$ for each $z_j = 0$. 
This will be enforced by using the variables $\summlit$.
In fact, except for blocks that are leaves, 
the only information about what literals are contained in
each block is given by these $\summlit$ variables. 
Since the layer of a block determines what variables can be present
in the block, only the relevant $\summlit$
variables are defined for each block. %
In particular, the root does not contain any $\summlit$
variable and the leaves contain all $\summlit$ variables. 
Moreover, the leaves are the only blocks that contain the $\lit$ variables, as
only then are these needed to connect the clause in the block with 
the axiom of $F$ that it is supposedly a weakening of.

We are now ready to introduce the variables of our formula.

\newlength{\mylength}
\settowidth{\mylength}{$\summlit_{i,z}^{B}$}
\newlength{\mylonglength}
\setlength{\mylonglength}{\glueexpr\textwidth - 18pt\relax}
\addtolength{\mylonglength}{-\mylength}

\begingroup
\setlength{\tabcolsep}{2pt} 
\def\arraystretch{1.3}
\begin{longtable}{>{\RaggedLeft}p{\mylength}  c p{\mylonglength} } 
$ \lit_\ell^{B}$ 
        &:&
        literal $\ell \in \Lit_n $ is not present in block $B \in \textsf{Leaves}$;  
  \\
   $ \weak_A^{B}$
        &: &
        block $B \in \textsf{Leaves}$ was obtained by weakening axiom $A \in [|F|]$;
    \\
   $ \weak\bin_j^B$
        &: &
        for $B \in \textsf{Leaves}$ and $j \in [\log |F|]$, the value of the $j$\textsuperscript{th} bit of 
        $\mathsf{bin}(A) \in \bit^{ \log |F|}$, where $A \in [|F|]$ is the axiom that block $B$ was
        obtained from by weakening;
    \\ 
$\summlit_{i,z}^{B} $
        &:&
        for block $B \in [s]$,
            $i \in [\layer(B)]$, 
            and $z \in \bit^{\noverdepth}$,
            for each $z_j = 1$, literal $x_{(i-1)\formuladepth + j}$ is not present in block $B$,
            and for each $z_j = 0$, literal $\lnot x_{(i-1)\formuladepth + j}$ is not present in block 
            $B$; \\
    $\point_{B',z}^{B}$
        &:&
            for blocks $B,B' \in [s+1]$ and $z \in \bit^{\noverdepth}$,
			the $z$\textsuperscript{th} child of block $B$ is $B'$ 
			(and therefore $\summlit_{i,z}^{B'} = 1$ for $i = \layer(B') = \layer(B)+1$);
		\\
    $\point\bin_{z,j}^{B}$
        &:&
        for block $B \in [s+1]$, $j \in [\log \delta]$,  and $z \in \bit^{\noverdepth}$, the value of the 
        $j$\textsuperscript{th} bit of $\mathsf{bin}(\sur_{z,\col(B)}(B')) \in \bit^{\log \delta}$ where $B'$ 
        is the $z$\textsuperscript{th} child of $B$
        (and, in particular, $\layer(B') = \layer(B) +1$);
         \\
        $\enable^{B}$
        &: &
        block $B \in [s+1]$ is enabled.
\end{longtable}

\endgroup

We record two remarks regarding the variables of the formula. First, we note a bound on the number of variables.

\begin{remark}\label{rem:var-count}
    For each class of variables listed above there are fewer than $O\left(|F| \cdot n \cdot s^2 \cdot 
    2^{\noverdepth}\right)$ variables of the class. Assuming $|F| = n^{O(1)}$, $d\leq n$, and $s = 
    2^{O(\noverdepth)}$, this is $2^{O(\noverdepth)} \cdot n^{O(1)}$.
\end{remark}

Second, we observe that the divisibility assumptions mentioned earlier can easily 
be dropped.

\begin{remark}\label{rk:divisibility}
We now explain why we can make the divisibility assumptions we made about $|F|,d,n$ and $\delta$.
\begin{itemize}
    \item 
    When $d$ does not divide $n$, we can add one final layer with a corresponding set of summary 
    and pointer variables to manage those $x_i$ with $i > \lfloor n/d\rfloor d$.
    \item When $|F|$ and $\delta$ are not powers of two, we can add clauses which rule out each invalid binary pointer or interpret each invalid pointer as some canonical value.
    
\end{itemize}
\end{remark}

\myparagraph{The clauses of $\eREF_d(F)$} 
We first introduce the unsplit versions (i.e., of width possibly larger than $3$) of the clauses which make up
$\eREF_d(F)$. 
These clauses enforce the basic structure of the Shallow-Resolution formula ensuring particularly 
that our summary and pointer variables interact correctly. The clauses are defined for arbitrary $G$ 
but in our proofs we focus on two cases for $G$: either $G$ is a bipartite expander with $t$ nodes 
on the left side, $t \cdot 2^{n/d}$ on the right side, and outdegree $\delta$ from each node the left 
side; or $G$ is a complete graph such that all $(\col(B), (\col(B'),z)) \in E(G)$.

We divide the clauses into three categories depending on whether the block in the superscript is in 
layer $0$ (meaning the block is the root), in layer $d$ (meaning the block is a leaf), or between 
layers $1$ and $d-1$ (meaning the block is an internal block).

\myparagraph{Root case --- final  block is enabled} 
    \begin{axiomdef}{2}{\textsc{Ref}}
    \setcounter{equation}{0} 
        & 
            \enable^{s+1} && 
            \label{axiom:root-enabled}
    \end{axiomdef}

\myparagraph{Leaf case --- managing axioms}
The following clauses are defined for each $B\in \mathsf{Leaf}$.
    \begin{axiomdef}{2}{\textsc{Ref}}
    \setcounter{equation}{1}
        & 
            \left( \enable^B \land \weak^B_A \right) \to \lnot \lit^B_\ell
        && \hspace{28pt} 
            \text{for } A\in[|F|],\
            \ell , \text{ in clause } A;
            \hspace{-40pt}  
            \label{axiom:must-appear-after-weak}
        \\
        & 
            \left( 
            \bigwedge_{j \in [\log |F|]} \left(\weak\bin^B_{j}\right)^{\mathsf{bin}(A)_j} \right) \to 
            \weak^B_{A} \quad
        && \hspace{28pt} 
            \text{for } A \in [|F|];\hspace{-40pt} 
            \label{axiom:weaken-binary-pointer} 
        \\
        & 
            \left(\enable^B \land \summlit_{i,z}^{B}\right)
            \to 
            \lit_{(x_{k})^{z_j}}^{B}  
            && \hspace{28pt} 
            \text{for } i\in[\formuladepth],  j \in [\noverdepth],\hspace{-40pt}  \label{axiom:summ-to-lits} 
            \\
            & && \hspace{28pt} 
            k = (i-1)\noverdepth+j \,. \hspace{-40pt}   \nonumber
    \end{axiomdef}

\myparagraph{Internal Step --- pointer and summary enforcement}
The following clauses are defined for each $B$ that is not a leaf, i.e., for $B$ such that $\layer(B) < \formuladepth$; for each $z\in \bit^{n/d}$; and for each $B'$ which is a possible $z$th child for $B$, i.e., for $B' \in P_B \coloneqq \{B': \layer(B') = \layer(B)+1$ and $((\col(B),z), \col(B')) \in E(G) \}$.
    \begin{axiomdef}{2}{\textsc{Ref}}
    \setcounter{equation}{4}
                & 
            \left( 
            \bigwedge_{j \in [\log \delta]} \left(\point\bin^B_{z,j} \right)^{\bin(\sur(B'))_j} \right) \to \point^B_{B',z} 
            \label{axiom:point-binary-pointer} 
        \\
        & 
            \left( \enable^B \land  \point^B_{B',z} \right) \to \enable^{B'}  
            \label{axiom:must-enable}
            \\
        & 
            \left(\enable^B \land \point^B_{B',z'} \land \summlit_{i,z}^{B} \right)  \to 
            \summlit_{i,z}^{B'} &&\hspace{25pt} \text{for } \text{$z' \in \bit^{\noverdepth}$,}\hspace{-30pt}
        \label{axiom:pass-on-summ-lits} \\
        &&&\hspace{25pt} \text{$i \le \layer(B)$;}  \nonumber\\
        & 
            \left( \enable^B \land \point^B_{B',z}\right) \to \summlit_{i+1,z}^{B'}
            &&\hspace{25pt} \text{for } i = \layer(B). \hspace{-30pt}
        \label{axiom:pointer-to-summ-lits} 
    \end{axiomdef}

This completes the list of the clauses of the non 3-CNF version of the formula $\eREF_d(F)$. 
We first observe that there are not too many clauses in this formula and that they are not too wide.

\begin{remark}\label{rem:clause-count}
    For each class of clauses listed above, there are fewer than $O(|F| \cdot n \cdot s^2 \cdot 2^{2n/d})$ clauses of that class, and each clause has fewer than $O(\log(|F|) +\log(\delta))$ literals in them.
\end{remark}

We now introduce a final set of variables which we refer to as ``split variables''. These are used to transform the formula into $3$-CNF using a standard transformation. For every clause $C_i = v^i_1 \lor v^i_2 \lor ... \lor v^i_{|C_i|}$ so far introduced, we introduce a collection of split variables $\splitv_{C_i}^j$ for each $j \in [|C_i| - 3]$ and replace $C_i$ with a collection of clauses $(v^i_1 \lor v^i_2 \lor \splitv_{C_i}^1) \land (\lnot \splitv_{C_i}^1 \lor v^i_3 \lor \splitv_{C_i}^2) \land ... \land (\lnot \splitv_{C_i}^{|C_i|-3} \lor v^i_{|C_i| - 1} \lor v^i_{|C_i|})$.

We denote the formula after splitting by $\eREF_d(F)$.
Note that, by \myref{rem:clause-count}, if $s = 2^{O(\noverdepth)}$, $d\leq n$, $|F| = n^{O(1)}$, and $\delta 
= O(s)$ then before splitting, the formula has $2^{O(\noverdepth)} \cdot n^{O(1)}$ 
clauses containing 
at most $O(\log n + \log s)$ variables each. As a consequence, after splitting, $\eREF_d(F)$ has at 
most $O(\log n + \log s)\cdot2^{O(\noverdepth)} \cdot n^{O(1)}= 2^{O(\noverdepth)} \cdot 
n^{O(1)}$ clauses, each of which contains at most $3$ variables. We record this observation below.

\begin{remark}\label{rem:clause-count-3}
The formula $\eREF_d(F)$ has at 
most $2^{O(\noverdepth)} \cdot 
n^{O(1)}$ clauses, each of which contains at most $3$ variables.
\end{remark}

\subsection{Main Theorem}
We prove our main theorem for an arbitrary expander graph, and we then instantiate this theorem with two different graphs. 

\begin{theorem}
    \label{lem:extended-ref}
    Given a $3$-CNF formula $F$ on $n$ variables, with $ \poly(n)$ clauses, and positive integers $d 
    \leq n$ and $t \leq 2^{\noverdepth}$, and a $(t\cdot 2^{n/d},2^{n/d},\delta,r,c)$-bipartite 
    expander graph
    $G$ there is a deterministic algorithm running in time $2^{O(\noverdepth)} \cdot n^{O(1)}$  that 
    constructs a $3$-CNF formula
    $\eREF_d(F)$ 
    which satisfies the following properties:
    \begin{enumerate}
        \item $\eREF_d(F)$ has $2^{\Theta({n/d})} \cdot n^{\Theta(1)}$ variables and 
        $2^{\Theta({n/d})} 
        \cdot n^{\Theta(1)}$ clauses.
        \item $\eREF_d(F)$ has width $O(1)$;
        \item $\eREF_d(F)$ admits Resolution proofs of depth at most $O(d\log \delta + \log n)$ when $F$ is satisfiable;
        \item $\eREF_d(F)$ requires Resolution proofs of width at least $2^{\Omega(r(c-1)/n)}$
        when $F$ is unsatisfiable. \label{it:thm-extended-ref-unsat}
\end{enumerate}
\end{theorem}

\begin{proof}

The first two claims follow directly from \myref{rem:clause-count-3}. 
The fourth claim follows from essentially the same argument given in
\cite{DBLP:conf/coco/CareniniR25} (reducing a hard for resolution pigeonhole principle to
finding a collision on the expander $G$), which we defer to 
\ifthenelse{\conferenceversion=1}{%
the full-length version of this paper.}{%
\myref{sec:lower-bound}.}

The third claim follows from the guaranteed success of the following
$O(d\log \delta + \log n)$ depth prover in the standard Prover-Delayer
game~(see \myref{sec:proof-complexity}) used to
analyse bounds for resolution proofs.
We lay out the strategy the prover takes in each of the three cases: root where $B=s+1$, internal where $B \leq s$ and $\layer(B) < d$, leaf where $\layer(B) = d$, and backchecking which is done after finding a leaf which is a valid weakening.
For each of these, we specify the types of clauses the prover could identify as being violated if the delayer gave answers that do not allow the prover to proceed in each step. In all cases it should be clear that the prover can identify a violated clause of 
one of the listed classes  \ref{axiom:root-enabled}--\ref{axiom:pointer-to-summ-lits}
immediately. 
Moreover, since all non-split clauses have at most $O(\log\delta + \log |F|) =O(\log\delta + \log n)$ variables, finding a split clause which is violated can be done in additional depth $O(\log\delta + \log n)$ by querying all the variables in the clause and all the split variables associated with the clause. 

The prover critically knows an assignment $x^*$ which satisfies $F$ and
will use this knowledge in implementing their strategy. In discussing the
strategy below, we apply subscripts to assignment strings $x^*_{i}$ and
summary strings $z^{j+1}_{i}$ with these subscripts serving as indices specifying the $i$\textsuperscript{th} bit in the string.

\myparagraph{Invariance condition:}

After applying the root case and $k$ internal steps of the strategy (but
not having yet applied the leaf case or the backchecking step), the following invariance conditions are maintained:
\begin{itemize}
    \item for all $j \in 0 \cup [k]$ the prover knows the variables $\enable^{B_j}$ and 
    $\point_{B_{j+1},z^{j+1}}^{B_j}$ which are all equal to $1$;
    \item $B_0 = s+1$;
    \item $\forall j \in [\formuladepth], \forall i \in [\noverdepth]: z^{j+1}_i = x^*_{j\noverdepth+i}$.
\end{itemize}

\myparagraph{Root Case:} 
Learn that $s+1$ is enabled by querying $\enable^{s+1}$ which will equal 1, otherwise 
\eqref{axiom:root-enabled} is violated. 
Learn the index $B_1$ of the child which is unsatisfied by $x^*$ by 
querying the $\log \delta$ variables $\point\bin^{s+1}_{z^1,j}$  where $z^{1}_{i} = x^*_{i}$.
Finally, query the variable $\point^{s+1}_{B_1,z}$ which must be $1$ otherwise
\eqref{axiom:point-binary-pointer} is violated.

\myparagraph{Internal Step:}
Learn that $B_j$ is enabled by querying $\enable^{B_j}$ which will equal 1 
\eqref{axiom:must-enable}. Learn the index $B_{j+1}$ of the child of $B_j$ which is unsatisfied by 
$x^*$ by querying the $\log \delta$  variables $\point\bin^{B_j}_{z^{j+1},j'}$  where $z^{j+1}_{i} = 
x^*_{j\noverdepth+i}$.
Finally,  query the variable  $\point^{B_j}_{B_{j+1},z}$ which must   be $1$ otherwise
right side of \eqref{axiom:point-binary-pointer}   is violated.

\myparagraph{Leaf Case:} 
Learn that $B_j$ is enabled by querying $\enable^{B_j}$ which will equal $1$, otherwise 
\eqref{axiom:must-enable} is violated. Query all $\summlit_{i,z^{i}}^{B_j}$ where each $z^i$ is 
specified 
according to the system in the internal step.
If they are all $1$ then, learn the axiom $A\in [|F|]$ which $B_j$ is a weakening of
by querying the $\log |F|$ variables $\weak\bin^{B_{j}}_{j'}$. 
Then query the variable
$\weak^{B_{j}}_A$ 
which must be $1$ otherwise it would violate \eqref{axiom:weaken-binary-pointer}.
Since $x^*$ is a satisfying assignment, one of \eqref{axiom:must-appear-after-weak} or 
\eqref{axiom:summ-to-lits} will be violated. 
We are left with the case where at least one of the $\summlit_{i,z^{i}}^{B_j}$  variables
is $0$.

\myparagraph{Backchecking:} 
For the indices $i$ and $j$ for which $\summlit_{i,z^{i}}^{B_j}$ was found to be $0$, query all
$\summlit_{i,z^{i}}^{B_\kappa}$ for $\kappa \in [i+1,j-1]$. If
$\summlit_{i,z^{i}}^{B_{i+1}} = 0$ then the appropriate
\eqref{axiom:pointer-to-summ-lits} is violated. Otherwise the prover finds a $\kappa$ such that
$\summlit_{i,z^{i}}^{B_{\kappa}} = 1$ but $\summlit_{i,z^{i}}^{B_{\kappa+1}} = 0$, which implies
the appropriate \eqref{axiom:pass-on-summ-lits} is violated. 

In the root case, and each of the $d$ internal steps, the prover queries $O(\log \delta)$ variables; in 
the leaf case, the prover queries $O(\log |F|)$ variables; and in the backchecking case, the prover 
queries at most $d$ variables. In total, it makes at most $O(d \log \delta + \log |F|)$ queries and is 
guaranteed to find a violated clause.
\qedhere

\end{proof}

From \myref{lem:extended-ref} we obtain two corollaries.
First, using the fact that an unbalanced complete bipartite graph is an
expander~(\myref{rk:expander}), we obtain the
following.
\begin{corollary}

    \label{lem:extended-ref-det}
    Given a 3-CNF formula $F$ on $n$ variables, with $ \poly(n)$ clauses, and positive integers $d 
    \leq n$ 
    and $t \leq 2^{\noverdepth}$, there is a deterministic algorithm running in time 
    $2^{O(\noverdepth)} \cdot n^{O(1)}$  that constructs a $3$-CNF formula
    $\eREF_d(F)$ 
    which satisfies the following properties":
    \begin{enumerate}
        \item $\eREF_d(F)$ has $2^{O({n/d})} \cdot n^{O(1)}$ variables and $2^{O({n/d})} 
        \cdot n^{O(1)}$ clauses;
        \item $\eREF_d(F)$ has width $O(1)$;
        \item $\eREF_d(F)$ admits Resolution proofs of depth at most $O(d\log t + \log n)$ when $F$ is satisfiable;
        \item $\eREF_d(F)$ requires Resolution proofs of width at least $2^{\Omega(t/n)}$
        when $F$ is unsatisfiable.
\end{enumerate}
\end{corollary}

Using a randomised construction of expanders~(\myref{lem:random-expander-constant}), we can improve the above construction in the
following way.
\begin{corollary}
    \label{lem:extended-ref-rand}
    Given a 3-CNF formula $F$ on $n$ variables, with $ \poly(n)$ clauses, and a positive integer 
    $d\leq n$, 
    there is a randomized algorithm running in time $2^{O(n/d)} \cdot
    n^{O(1)}$  that
    constructs a 
    $3$-CNF formula
    $\eREF_d(F)$ 
    which, with high probability, satisfies the following properties:
    \begin{enumerate}
        \item $\eREF_d(F)$ has $2^{O({n/d})} \cdot n^{O(1)}$ variables and $2^{O({n/d})} 
        \cdot n^{O(1)}$ clauses;
        \item $\eREF_d(F)$ has width $O(1)$;
        \item $\eREF_d(F)$ admits Resolution proofs of depth at most $O(d + \log n)$ when $F$ is satisfiable;
        \item $\eREF_d(F)$ requires Resolution proofs of width at least $2^{2^{\Omega(n/d)}}$
        when $F$ is unsatisfiable.
\end{enumerate}
\end{corollary}

\section{Learning Lower Bounds Under ETH}

In this section, we formally 
prove \cref{thm:introthm1,thm:introthm2}.
We begin by defining a gap learning decision problem
to which we  reduce the SAT problem,
and from which we are able to reduce to other learning problems.

\begin{definition}
    \label{def:mcgl}
    Let $s_1, s_2, \eps : \bbn \to \bbn$
    be time-constructible functions.
    The computational problem
    $\mCFGL_{s_1}^{s_2}[\eps]$ 
    (referred to as \emph{Monotone Circuit-Formula Gap Learning})
    receives as input
    $(1^n, \cald)$,
    where $\cald$ is a circuit with
    $n+1$ output bits and size $n^2$
    (here, we identify the circuit with its representation for simplicity of
    notation).
    We interpret $\cald$ as a distribution
    of pairs $(x,b) \in \blt^n \times \blt$,
    obtained by sampling a uniformly random input $r$
    and outputting $\cald(r)$.
    The task is to decide
    if there exists a monotone 
    formula
    $C$ 
    (with $n$ input bits and single output bit)
    of size $s_1$
    such that $\Pr_{(x,b) \gets \cald}[C(x)=b]=1$,
    or if every monotone circuit $C'$
    such that $\Pr_{(x,b) \gets \cald}[C'(x)=b] \geq 1/2+\eps$
    must have size at least $s_2$.
    The number $n$ is called the \emph{dimension} of the input.
\end{definition}

We record an observation about the sampling circuit that we use later on.
\begin{remark}
    \label{rk:mcgl-def}
    We observe that the size of the sampling circuit $\cald$ is chosen to be
    $n^2$ arbitrarily but without loss of
    generality. Indeed, given any family of inputs where the $n$-bit
    samplers have size $s = n^{O(1)}$, if $s < n^2$ we can just add dummy gates to the
    circuit; if $s > n^2$ we can pad the output $x$ with $\sqrt{s}-n$ 0's;
    in both cases, we reduce to the case $s = n^2$ in polynomial-time.
\end{remark}

We now implement the strategy outlined in the introduction.
We take a 3-CNF formula~$F$ on $n$ variables, construct the formula $\eREF_d(F)$
satisfying either \myref{lem:extended-ref-det} 
(in case of a deterministic reduction)
or \myref{lem:extended-ref-rand} 
(in case of a randomised reduction).
Using a deterministic reduction,
taking $m$ such that $n \leq m \ll 2^{n/d}$, $t \leq 2^{n/d}$
and $\ell = \eps dm/n \ll 2^{t/n}$ for small enough $\eps$ in 
\myref{thm:constructive-lifting} gives that the
function $f_{\eREF_s(F),m}$ has
$N = 2^{O(n/d)} \cdot n^{O(1)}$ input bits
and the following monotone complexity:
\begin{enumerate}
    \item At most $m^{{O}(d\log t + \log n)} = 2^{O( (d\log t+\log n) \log m )} $
    when $F$ is satisfiable, and
    \item At least $2^{\Omega(\ell)} \geq 2^{\Omega(m/\log N)} $
    when
        $F$ is unsatisfiable.
\end{enumerate}
Using a randomised reduction, the value of $t$ is irrelevant
and
the resulting monotone complexity is
\begin{enumerate}
    \item At most $m^{{O}(d + \log n)} = 2^{O( (d+\log n) \log m )} $
    when $F$ is satisfiable, and
    \item At least $2^{\Omega(\ell)} \geq 2^{\Omega(m/\log N)} $
    when
        $F$ is unsatisfiable.
\end{enumerate}
Since we are reducing from ETH, the value $m = \poly(n)$
is already good enough, as a monotone circuit lower bound beyond $2^{\Omega(n)}$
(for an $N$-variate function)
is not of much use for reductions to Partial-MCSP or PAC-learning (since we would need to
output a circuit of this size or generate an exponential number of examples).
The value $d$ is then the main parameter we choose. There are two parameter regimes we
will care about:
\begin{enumerate}
    \item \emph{($d = \log n$)}. 
        We obtain a formula upper bound
        which is smaller than the number $N$ of input bits (a ``junta'').
        The runtime lower bound based on ETH becomes
        $N^{\Omega(\log \log N)}$.
    \item \emph{($d \approx \sqrt{n}$)}.
        This gives us a polynomial vs.\ quasipolynomial gap.
        The runtime lower bound becomes
        $N^{\tld{\Omega}(\log N)}$.
\end{enumerate}

\subsection{ETH-hardness of Monotone Circuit-Formula Gap Learning}

We first show the result for juntas.

\begin{theorem}[Junta Gap-Learning]
    \label{thm:sat-reduces-to-mcgl}
    Let $\gamma > 0$ be a constant.
    For large enough $n \in \bbn$,
    the 3SAT problem
    on $n$ variables
    reduces to
    $\mCFGL_{s_1}^{s_2}[\gamma]$
    in time
    $2^{O(n / \log n)}$,
    where
    $s_1(N) = 2^{O(\log \log N)^3}$
    and
    $s_2(N) = N^{\Omega(\log N)}$.
    Consequently, the problem
    $\mCFGL_{s_1}^{s_2}[\gamma]$
    needs time $N^{\Omega(\log \log N)}$ to be solved under ETH.
\end{theorem}
\begin{proof}
Let $F$ be a given 3-CNF over $n$ variables.
Fix $d = \log n$.
Note that $\eREF_s(F)$ is a formula of width $O(1)$ over
$N = 2^{\Theta(n/\log n)}$
variables and with $N^{O(1)}$ clauses.
Let
\begin{equation*}
    m := n^3/(\log n)^3,
    \quad
    t := n^2.
\end{equation*}
and let $\ell = \eps dm/n$ for a sufficiently small constant $\eps > 0$.
The function $\fT := f_{\eREF_d(F), m}$
is defined on 
$n_0 := N^{O(1)} m^{O(1)} = 2^{O(n/\log n)}$ variables
(\myref{def:function-lifting}).
\cref{thm:constructive-lifting,lem:extended-ref}
imply that
$\fT$
can be computed
with monotone formulas of size
$s := m^{O(d \log t + \log n)} = 2^{O(\log^3 n)} = 2^{O(\log \log N)^3}$
when $F$ is satisfiable.
Moreover, 
when $F$ is unsatisfiable,
the function
requires
monotone circuits of size at least 
\begin{equation*}
    s' :=
    2^{\Omega(\ell)}
    =
    N^{\Omega(\log N)}
\end{equation*}
to be $\gamma$-approximated over the distribution 
$\cald :=
\cald^{\eREF_s(F),m}$
given by \myref{thm:constructive-lifting}, 
since we can take $\eps$ sufficiently smaller than $\gamma$.
This distribution can be constructed in time $2^{O(n/\log n)}$ given
$F$ and $m$.
A postprocessing taking $2^{O(n/\log n)}$-time may be needed to
adjust the description of the circuit to the required length (see
\myref{rk:mcgl-def}).
The ETH lower bound follows by noticing that
$N^{\log \log N} = 2^{O(n)}$.
\end{proof}

The result for polynomial vs.\ quasipolynomial gap follows similarly.

\begin{theorem}[Polynomial vs.\ quasipolynomial gap]
    Let $\gamma > 0$ be a constant.
    For large enough $n \in \bbn$,
    the 3SAT problem
    on $n$ variables
    reduces to
    $\mCFGL_{s_1}^{s_2}[\gamma]$
    in time
    $2^{O(\sqrt{n} \log n)}$,
    where
    $s_1(N) = N$
    and
    $s_2(N) = N^{\Omega(\log N)}$.
    Consequently, the problem
    $\mCFGL_{s_1}^{s_2}[\gamma]$
    needs time $N^{\Omega( (\log N)/(\log \log N)^2 )}$ to be solved under ETH.
    \label{thm:sat-reduces-to-mcgl-quasipoly}
\end{theorem}
\begin{proof}
Let $F$ be a given 3-CNF over $n$ variables.
Fix $d = \alpha \sqrt{n}/\log n$ where $\alpha > 0$ is small enough constant.
Note that $\eREF_s(F)$ is a formula of width $O(1)$ over
$N = 2^{\Theta(\sqrt{n} \log n)}$ 
variables and with $N^{O(1)}$ clauses.
Let
\begin{equation*}
    m := (\log N)^2 \asymp n \log^2 n,
    \quad
    t := n^2.
\end{equation*}
and let $\ell = \eps dm/n$ for a sufficiently small constant $\eps > 0$.
The function $\fT := f_{\eREF_d(F), m}$
is defined on 
$n_0 := N^{\Theta(1)} m^{\Theta(1)} = 2^{\Theta(\sqrt{n} \log n)} = N^{\Theta(1)}$
variables
(\myref{def:function-lifting}).
\cref{thm:constructive-lifting,lem:extended-ref}
imply that
$\fT$
can be computed
with monotone formulas of size
\begin{equation*}
    s := m^{O(d \log t + \log n)} = 2^{O(d \log^2 n + \log^2 n)} = 2^{O(\alpha \sqrt{n} \log n)}
    \leq n_0,
\end{equation*}
when $F$ is satisfiable and $\alpha$ is small enough.
Moreover, 
when $F$ is unsatisfiable,
the function
requires
monotone circuits of size at least 
\begin{equation*}
    s' :=
    2^{\Omega(\ell)}
    =
    N^{\Omega(\log N)}
\end{equation*}
to be $\gamma$-approximated over the distribution 
$\cald :=
\cald^{\eREF_s(F),m}$
given by \myref{thm:constructive-lifting}, 
since we can take $\eps$ sufficiently smaller than $\gamma$.
This distribution can be constructed in time $2^{O(\sqrt{n} \log n)}$ given
$F$ and $m$.
The ETH lower bound follows by noting that
$N^{(\log N)/(\log\log N)^2} = 2^{O(n)}$.
\end{proof}

\subsection{rETH-hardness of Monotone Circuit-Formula Gap Learning}
\label{sec:reth}

Next, we show slightly stronger versions of these results but under rETH.

\begin{theorem}[Junta Gap-Learning under rETH]
    \label{thm:sat-reduces-to-mcgl-junta-reth}
    Let $\gamma > 0$ be a constant.
    For large enough $n \in \bbn$,
    the 3SAT problem
    on $n$ variables
    randomly reduces to
    $\mCFGL_{s_1}^{s_2}[\gamma]$
    in time
    $2^{O(n / \log n)}$,
    where
    $s_1(N) = (\log N)^{O(\log \log N)}$
    and
    $s_2(N) = N^{\Omega(\log N)}$.
    Consequently, the problem
    $\mCFGL_{s_1}^{s_2}[\gamma]$
    needs time $N^{\Omega(\log \log N)}$ to be solved under rETH.
\end{theorem}
\begin{proof}
Let $F$ be a given 3-CNF over $n$ variables.
Fix $d = \log n$.
Note that $\eREF_s(F)$ is a formula of width $O(1)$ over
$N = 2^{\Theta(n/\log n)}$ 
variables and with $N^{O(1)}$ clauses.
Let
\begin{equation*}
    m := n^3/(\log n)^3,
\end{equation*}
and let $\ell = \eps dm/n$ for a sufficiently small constant $\eps > 0$.
The function $\fT := f_{\eREF_d(F), m}$
is defined on 
$n_0 := N^{O(1)} m^{O(1)} = 2^{O(n/\log n)}$ variables
(\myref{def:function-lifting}).
\cref{thm:constructive-lifting,lem:extended-ref}
imply that
$\fT$
can be computed
with monotone formulas of size
$s := m^{O(d + \log n)} = 2^{O(\log^2 n)} = 2^{O(\log \log N)^2}$.
when $F$ is satisfiable.
The rest of the proof is the same as in \myref{thm:sat-reduces-to-mcgl}.
\end{proof}

\begin{theorem}[Polynomial vs.\ quasipolynomial gap under rETH]
    \label{thm:sat-reduces-to-mcgl-quasipoly-reth}
    Let $\gamma > 0$ be a constant.
    For large enough $n \in \bbn$,
    the 3SAT problem
    on $n$ variables
    reduces to
    $\mCFGL_{s_1}^{s_2}[\gamma]$
    in time
    $2^{O(\sqrt{n})}$,
    where
    $s_1(N) = N$
    and
    $s_2(N) = N^{\Omega(\log N)}$.
    Consequently, the problem
    $\mCFGL_{s_1}^{s_2}[\gamma]$
    needs time $N^{\Omega(\log N)}$ to be solved under rETH.
\end{theorem}
\begin{proof}
Let $F$ be a given 3-CNF over $n$ variables.
Fix $d = \alpha \sqrt{n}$ where $\alpha > 0$ is small enough constant.
Note that $\eREF_s(F)$ is a formula of width $O(1)$ over
$N = 2^{\Theta(\sqrt{n})}$ 
variables and with $N^{O(1)}$ clauses.
Let
\begin{equation*}
    m := (\log N)^2 \asymp n,
\end{equation*}
and let $\ell = \eps dm/n$ for a sufficiently small constant $\eps > 0$.
The function $\fT := f_{\eREF_d(F), m}$
is defined on 
$n_0 := N^{\Theta(1)} m^{\Theta(1)} = 2^{\Theta(\sqrt{n} \log n)}$ variables
(\myref{def:function-lifting}).
\cref{thm:constructive-lifting,lem:extended-ref}
imply that
$\fT$
can be computed
with monotone formulas of size
\begin{equation*}
    s := m^{O(d  + \log n)} = 2^{O(d)} = 2^{O(\alpha \sqrt{n})}
    \leq n_0,
\end{equation*}
when $F$ is satisfiable and $\alpha$ is small enough.
Moreover, 
when $F$ is unsatisfiable,
the function
requires
monotone circuits of size at least 
\begin{equation*}
    s' :=
    2^{\Omega(\ell)}
    =
    N^{\Omega(\log N)}
\end{equation*}
to be $\gamma$-approximated over the distribution 
$\cald :=
\cald^{\eREF_s(F),m}$
given by \myref{thm:constructive-lifting}, 
since we can take $\eps$ sufficiently smaller than $\gamma$.
This distribution can be constructed in time $2^{O(\sqrt{n})}$ given
$F$ and $m$.
The ETH lower bound follows by noticing that
$N^{\log N} = 2^{O(n)}$.
\end{proof}

\subsection{From Gap Learning to Partial-Gap-MCFSP}

We now define the partial MCSP problem we consider.
Instead of the entire truth table of a partial function,
we consider a more succinct encoding where we are only given locations
where the function is defined.

\begin{definition}
    Let $s_1,s_2,\eps : \bbn \to \bbn$ be time-constructible functions.
    Let $\mMCFSP_{s_1}^{s_2}[\eps]$ be the promise problem which receives as input
    $1^n$ and
    a list of $m$ pairs
    $((x_1,b_1),\dots,(x_m,b_m))$
(called \emph{examples}),
where $x_i \in \blt^n$ and $b_i \in \blt$ for all $i \in [m]$,
and must distinguish between the following two cases:
\begin{itemize}
    \item 
    there is a monotone formula $C$ of size at most $s_1$ such that, for all $i \in [m]$, $C(x_i) = b_i$; 
    \item 
    for all circuits $C'$ of size less than $s_2$,
    we have
    $\card{\set{i \in [m] : C'(x_i) \neq b_i}} \geq m(1/2-\eps)$.
\end{itemize}
    The number $n$ is called the \emph{dimension} of the input.
\end{definition}

To show hardness of $\mMCFSP$, it suffices to reduce from
$\mCFGL$.
As described in the introduction, it suffices to sample from
the distribution $\cald$ given as input to $\mCFGL$.
In the small formula case, we clearly generate examples
which can be computed by a small formula.
In the large circuit case, with enough examples we can conclude
by the Chernoff inequality that no circuit of small size
can compute most of those examples correctly.

\begin{theorem}
    \label{thm:monotone-minlt-eth-hard}
    For every 
    growing functions
    $s_1,s_2: \bbn \to \bbn$
    such that
    $s_1 \leq s_1' < s_2' \leq s_2$ 
    (for large enough $n$)
    and 
    $\gamma : \bbn \to [0,1]$,
    the problem
    $\mCFGL_{s_1}^{s_2}[\gamma]$
    with input $(1^n, \cald)$
    randomly reduces to
    $\mMCFSP_{s_1'}^{s_2'}[2\gamma]$
    in time
    $\poly(n) + O(\frac{1}{\gamma^2} \cdot s_2' \log s_2')$,
    producing 
    $O(\frac{1}{\gamma^2} \cdot s_2' \log s_2')$ examples
    of dimension $n$.
\end{theorem}
\begin{proof}
    Let $\cald$ be given as input for 
    $\mCFGL_{s_1}^{s_2}[\gamma]$.
    Take $m = \frac{4}{\gamma^2} \cdot K s_2' \log s_2'$ samples from $\cald$,
    and call them $(x_1,b_1),...,(x_m,b_m)$. 

    In the case that there exists a formula of size $s_1 < s_1'$ computing $\cald$ then
    that same small circuit will agree with $(x_1,b_1),...,(x_m,b_m)$ for
    all $i$. So ``small'' instances map to ``small'' instances. It will then
    suffice to prove the theorem to show that ``large'' instances map to
    ``large'' instances. In particular, it suffices to show that, with 
    all but negligible probability,
    for every circuit of $C'$ of size less than 
    $s_2'$,
    we have
    $\card{\set{i \in [m] : C'(x_i) \neq b_i}} \geq m(1/2-\eps-\gamma)$.

    Let $a_{C'}$ be the number of $i \in [m]$ such that $C'(x_i) = b_i$. 
    By definition of $\mCFGL$, we know
    $\Pr_{(x,b) \flws \cald}[C'(x) = b] \leq 1/2+\gamma$,
    and thus
    $\Exp[a_{C'}] \leq m(1/2 + \gamma)$.
    Therefore, for each
    circuit $C'$, by Chernoff-Hoeffding's inequality 
    we know that
    \begin{equation}
        \label{eq:err-bound}
    \Pr_{(x_1,b_1),...,(x_m,b_m) \leftarrow \cald}[a_{C'} - m(1/2 + \gamma)
    \geq m\gamma]
    \leq e^{-2m^2\gamma^2/m} = e^{-2m\gamma^2} 
    =
    (s_2')^{-4 s_2'}\,,
    \end{equation}
    or, equivalently,
    \begin{equation*}
        \Pr_{(x_1,b_1),...,(x_m,b_m) \leftarrow \cald}[a_{C'}\geq m(1/2 +
        2 \gamma)] \leq (s_2')^{-4 s_2'} \,.
    \end{equation*}

    There are less than 
    $(s_2')^{3s_2'}$
    monotone circuits with $n$
    bit inputs, $s_2'$ gates, fan-in $2$, and arbitrary fan-out. 
    By union bounding over these circuits, we can see that
    the probability of any circuit satisfying $a_{C'} \geq
    m(1+2\gamma)$ is less than $1/3$ for large enough $n$.
    This completes the proof.
\end{proof}

We now employ the hardness results obtained for $\mCFGL$
under rETH in \myref{sec:reth}
to obtain hardness of $\mMCFSP_{s_1}^{s_2}$ using the reduction of the previous
theorem.
We consider 4 choice of parameters which are in a way best possible:
\begin{enumerate}
    \item \emph{Junta vs.\ quasipolynomial}. Here we take
        the largest gap possible between formula and circuit size.
        This gives us a weaker runtime lower bound of
        $N^{\Omega(\log \log N)^{1-\eps}}$.
        This is the weakest runtime lower bound we obtain.
    \item \emph{Linear vs.\ quasipolynomial.}
        Here we obtain a weaker gap but against superpolynomial-size
        circuits.
        This gives us a better runtime lower bound of
        $N^{\Omega(\log N)^{1-\eps}}$.
    \item \emph{Junta vs.\ polynomial.}
        We then consider the setting where 
        the gap is largest possible subject to 
        $s_2 = \poly(n)$.
        This gives us a runtime lower bound of
        $N^{\Omega(\log \log N)}$, slightly better than (1).
    \item \emph{Linear vs.\ polynomial.}
        We weaken the gap to an arbitrary polynomial $n^c$.
        This gives us a better runtime lower bound of
        $N^{\Omega(\log N)}$, slightly better than (2).
        This is the strongest runtime lower bound we obtain.
\end{enumerate}

We obtain the following result for each of these parameter regimes.

\begin{corollary}
    \label{cor:minlt-runtime}
    Let $\gamma > 0$ be a constant.
    Under rETH,
    the following holds:
    \begin{enumerate}
        \item 
            \textbf{(Junta vs.\ quasipolynomial)}
            For every constant $\eps > 0$,
            the problem
            $\mMCFSP_{s_1}^{s_2}[\gamma]$
            requires time
            $N^{\Omega(\log \log N)^\eps}$
            to be decided,
            where 
            $s_1(n) = (\log N)^{\log \log N}$
            and
            $s_2(n) = n^{(\log \log n)^{1-\eps}}$.
            The hard instances consist of
            $m = \Theta(s_2 \log s_2)$ examples,
            and thus have size
            $N = nm = n^{\Theta(\log \log n)^{1-\eps}}$. \label{it:junta-qp}
        \item 
            \textbf{(Linear vs.\ quasipolynomial)}
            For every 
            constant $\eps \in (0,1)$,
            the problem
            $\mMCFSP_{s_1}^{s_2}[\gamma]$
            requires time
            $N^{\Omega((\log N)^{2/(2-\eps)})}$
            to be decided,
            where 
            $s_1(n) = n$,
            and 
            $s_2(n) = n^{(\log n)^{1-\eps}}$.
            The hard instances consist of
            $m = \Theta(s_2 \log s_2)$ examples,
            and thus have size
            $N = nm = n^{\Theta(\log n)^{1-\eps}}$. \label{it:linear-qp}
        \item 
            \textbf{(Junta vs.\ polynomial)}
            For every constant $\eps > 0$,
            the problem
            $\mMCFSP_{s_1}^{s_2}[\gamma]$ 
            requires time
            $N^{\Omega(\log \log N)}$
            to be decided,
            where 
            $s_1(n) = (\log N)^{\log \log N}$
            and
            $s_2(n) = n^{c}$.
            The hard instances consist of
            $m = \Theta(s_2 \log s_2)$ examples,
            and thus have size
            $N = \Theta(n m) = \poly(n)$. \label{it:junta-poly}
        \item 
            \textbf{(Linear vs.\ polynomial)}
            For every constant $c > 0$,
            the problem
            $\mMCFSP_{s_1}^{s_2}[\gamma]$ 
            requires time
            $N^{\Omega( \log N ) }$
            to be decided,
            where 
            $s_1(n) = n$,
            and
            $s_2(n) = n^{c}$.
            The hard instances consist of
            $m = \Theta(s_2 \log s_2)$ examples,
            and thus have size
            $N = \Theta(n m) = \poly(n)$. \label{it:linear-poly}
    \end{enumerate}
\end{corollary}
\begin{proof}
            We first prove \myref{it:junta-qp}.
            By \myref{thm:sat-reduces-to-mcgl-junta-reth}, for large enough
            $n \in \bbn$
            the problem 
            $\mCFGL_{s_1}^{s_2}[\gamma/2]$ 
            requires time $n^{\Omega(\log n)}$ to be computed
            under ETH for some $s_1 = 2^{O(\log \log n)^2}$
            and $s_2 = n^{\Omega(\log n)}$.
            Take $s_1' = s_1$
            and $s_2' = n^{(\log \log n)^{1-\eps}}$,
            thus satisfying
            $s_1 < s_1' < s_2' < s_2$.

            If $\mMCFSP_{s_1'}^{s_2'}[\gamma]$ can be computed in time
            $T(N)$ on an instance of $m = K s_2' \log s_2' = n^{O(\log \log n)^{1-\eps}}$
            examples and dimension $n$ 
            (and thus input size $N=nm$),
            where $K$ is a large enough constant,
            we obtain by 
            \myref{thm:monotone-minlt-eth-hard}
            that $\mCFGL_{s_1}^{s_2}[\gamma/2]$
            can be solved in randomised time
            $\poly(n) + O(m) + T(N)$.
            This must be at least 
            $n^{\Omega(\log \log n)}$
            which implies
            $T(N) = n^{\Omega(\log \log n)}$.
            Since 
            $N = n^{\Theta(\log n \log n)^{1-\eps}}$,
            we obtain 
            $\log n \log n = (1-o(1)) \log \log N$
            and
            \begin{equation*}
                T(N) \geq N^{\Omega(\log \log N)^{\eps}}.
            \end{equation*}

            We now prove \myref{it:linear-qp}.
            The proof is similar.
            Let
            $s_1' = n$ 
            and
            $s_2' = n^{(\log n)^{1-\eps}}$ in the argument above.
            We have that
            $\mCFGL_{s_1}^{s_2}[\gamma/2]$ 
            requires time $n^{\Omega(\log n)}$
            to be computed by \myref{thm:sat-reduces-to-mcgl-quasipoly-reth}.
            Since 
            $N \asymp n s_2' \log s_2' = n^{\Theta(\log n)^{1-\eps}}$,
            we obtain
            $\log n \asymp (\log N)^{1/(2-\eps)}$
            and thus
            \begin{equation*}
                T(N)
                \geq
                2^{\Omega(\log n)^2}
                \geq
                2^{\Omega(\log N)^{2/(2-\eps)}}.
            \end{equation*}

            We repeat the argument for \myref{it:junta-poly}.
            Take $s_1' = (\log n)^{\log \log n}$
            and
            $s_2' = n^c$.
            We have that
            $\mCFGL_{s_1}^{s_2}[\gamma/2]$ 
            requires time $n^{\Omega(\log \log n)}$
            to be computed by \myref{thm:sat-reduces-to-mcgl-junta-reth}.
            Since 
            $N \asymp n s_2' \log s_2' \asymp n^{c+1} \log n$,
            we obtain
            $\log N \asymp \log n$
            and
            $\log \log N \sim \log \log n$,
            and thus
            \begin{equation*}
                T(N)
                \geq
                n^{\Omega(\log \log n)}
                \geq
                N^{\Omega(\log \log N)}.
            \end{equation*}

            Finally, we prove \myref{it:linear-poly}.
            Let $s_1' = n$
            and let $s_2' = n^c$ in the argument above.
            We have that
            $\mCFGL_{s_1}^{s_2}[\gamma/2]$ 
            requires time $n^{\Omega(\log n)}$
            to be computed by \myref{thm:sat-reduces-to-mcgl-quasipoly-reth}.
            Since 
            $N \asymp n s_2' \log s_2' \asymp n^{c+1} \log n$,
            we obtain
            $\log N \asymp \log n$
            and thus
            \begin{equation*}
                T(N)
                \geq
                2^{\Omega(\log n)^2}
                \geq
                2^{\Omega(\log N)^2}.
                \qedhere
            \end{equation*}
\end{proof}

\subsection{From Gap Learning to Improper PAC-Learning}
\label{sec:improperpac}

Finally, we explain how to obtain hardness of PAC-learning. We first define formally the notion of PAC-learning.
Given a function $f : \blt^n \to \blt$ and a distribution 
$\cald$ with support in $\blt^n$,
we denote by $\Ex(f,\cald)$
the oracle that runs in unit time and outputs
a sample $(x,b)$ where $x \flws \cald$ and $b = f(x)$.
We also define
$\err_{f, \cald}(h) := \Pr_{x \in \cald}[h(x) \neq f(x)]$.

\begin{definition}
    \label{def:pac-learning}
    Let $\calf_n, \calh_n$ be sets of Boolean functions over $\blt^n$,
    and let 
    $\calf = \bigcup_{n \geq 1} \calf_n$
    and
    $\calh = \bigcup_{n \geq 1} \calh_n$
    be such that 
    $\calf \sseq \calh$.
    Let $T : \bbn \to \bbn$.
    We say that the \emph{concept class} $\calf$
    is
    \emph{PAC-learnable in time $T$} 
    by the 
    \emph{hypothesis class}
    $\calh$
    if there is a randomised algorithm such that,
    for every $f \in \calf$,
    given $(1^n, 1^{1/\eps}, 1^{1/\delta})$
    and oracle access to
    $\Ex(f,\cald)$,
    outputs 
    in time $T(n + 1/\eps + 1/\delta)$,
    with probability $1-\delta$,
    a concept $h \in \calh$
    such that
    $\err_{f, \cald}(h) \leq \eps$.
\end{definition}

For $s : \bbn \to \bbn$,
let $\calc_{s(n)}$ (resp. $\call_{s(n)}$) 
be the set of monotone Boolean circuits (resp. formulas) on $n$ bits of size $s(n)$.
Let 
$\calc_s = \bigcup_{n \in \bbn}\calc_{s(n)}$
and
$\call_s = \bigcup_{n \in \bbn}\call_{s(n)}$.
Hardness of PAC-learning $\call_s$ by $\calc_s$ 
follows from the results of \myref{sec:reth}
by a proof similar to the one employed in the previous section.

\begin{theorem}
    \label{thm:pac-lb} 
        Under rETH,
        we have: 
        \begin{enumerate}
            \item 
                For every constant $\alpha > 0$,
                monotone formulas on $n$ bits of size $n$
                require time $n^{\Omega( \log n )}$
                to be PAC-learned by
                monotone circuits of size 
                $s(n)$,
                for any $s(n)$ such that
                $n \ll s(n) \leq n^{(\log n)^{1-\alpha}}$.
            \item
                For every constant $\alpha > 0$,
                monotone formulas 
                on $n$ bits
                of size $(\log n)^{O(\log \log n)}$
                require time $n^{\Omega( \log \log n )}$
                to be PAC-learned by
                monotone circuits of size 
                $s(n)$,
                for any $s(n)$ such that
                $n \ll s(n) \leq n^{(\log \log n)^{1-\alpha}}$.
        \end{enumerate}
\end{theorem}
\begin{proof}
    We first prove the first item. The proof of the second item is analogous, so we omit
    it.
    
    Let $\eps = \delta = 1/10$.
    Let $(1^n, \cald)$ be an input to
    $\mCFGL_{s_1}^{s_2}[\eps]$,
    where $s_1(n) = n$ and $s_2(n)=n^{\Omega(\log n)}$ are as in
    \myref{thm:sat-reduces-to-mcgl-quasipoly-reth}. By this theorem, 
    we have that $\mCFGL_{s_1}^{s_2}[\eps]$ requires time $n^{\Omega(\log n)}$ to be solved under rETH.
    
    Suppose that $\call_n$ is PAC-learnable in time $T$ by $\calc_{s}$,
    that is, suppose 
    there is a PAC-learner $\cala$ that for every $f\in \call_n$, given oracle access to $\Ex(f,\cald)$,
    outputs 
    in time  $T = T(n + 1/\eps + 1/\delta)$
    a hypothesis $h$ in $\calc_{s}$ 
    that,
    with probability $1-\delta$, satisfies $\err_{f, \cald}(h) \leq \eps$.
    Run the algorithm $\cala$, replacing oracle calls to $\Ex$
    with samples from~$\cald$, and let $C$ be the circuit obtained at the end. 
    Since $\card{\cald} = \poly(n)$,
    this can be done in time $T \cdot \poly(n)$.
    We now take 
    $\ell=\frac{4}{\eps^2}$
    samples from $\cald$,
    and accept $\cald$ if $C$ computes correctly a $2/3$-fraction of the examples.
    
    We claim that in the 
    case where there exists a monotone formula of size $s_1$
    computing all samples of $\cald$ correctly, we will be correct with probability at least $1-\delta - \exp({-2\ell\eps^2})\geq 2/3$. Indeed, in this case, with probability $1-\delta$ it holds that $\err_{c, \cald}(C) \leq \eps$. Now suppose $C$ is such that $\err_{c, \cald}(C) \leq \eps$. By Chernoff-Hoeffding's inequality we have that
    the probability that $C$ does not compute correctly $1-2\eps \geq 2/3$ of the $\ell$ examples is at most $\exp(-2\ell\eps^2)$. The claim follows.
    
    In the other case, we know that 
    $\Pr_{(x,b) \flws \cald}[C(x) = b] \leq 1/2+\eps$,
    and thus
    the proportion of 
    samples 
    that $C$ correctly computes is
    less than $1/2+2\eps$ with 
    probability $1-\exp(-2 \ell \eps^2)$
    by Chernoff-Hoeffding's inequality.
    The total runtime is $T\cdot\poly(n) + \ell \cdot \poly(s)$. This implies that under rETH, by \myref{thm:sat-reduces-to-mcgl-quasipoly-reth}, 
    $T = n^{\Omega( \log n )}$.
\end{proof}

\subsection*{Acknowledgements}
We thank Noel Arteche for helpful discussions, and the
anonymous reviewers for their valuable comments and suggestions.
    Rahul Santhanam and Bruno Cavalar
    acknowledge support from the EPSRC project
    EP/Z534158/1 on ``Integrated Approach to Computational Complexity: Structure, Self-Reference and Lower
    Bounds''.
    Bruno Cavalar also acknowledges support of
    the UK Research and Innovation (UKRI) through an EPSRC
    Postdoctoral Fellowship (grant ref.\! UKRI3441).
    Susanna F. de Rezende
    was supported by the Knut and Alice Wallenberg Foundation grant
    KAW 2023.0116, ELLIIT, and the Swedish Research Council grant 2021-05104.

\let\OLDthebibliography\thebibliography
\renewcommand\thebibliography[1]{
   \OLDthebibliography{#1}
   \setlength{\parskip}{3pt}
}

\small
\DeclareUrlCommand{\Doi}{\urlstyle{sf}}
\renewcommand{\path}[1]{\footnotesize\Doi{#1}}
\renewcommand{\url}[1]{\href{#1}{\small\Doi{#1}}}
\bibliographystyle{alphaurl}
\bibliography{refs}

\appendix

\section{Proof Sketch of the Refined Lifting Theorem}
\label{sec:refined-lifting}
\normalsize

Throughout this appendix, we assume familiarity with the notation and definitions
from the paper~\cite{DBLP:conf/coco/RezendeV25}\footnote{In particular, we refer to the
\href{https://derezende.github.io/index_files/colourful-sunflowers.pdf}{full version} of
the paper.}.
\myref{thm:constructive-lifting} is proved
in~\cite{DBLP:conf/coco/RezendeV25} only for
the case when $\gamma = 1/2$.
However, it is straightforward to extend it to
the case $\gamma \geq 2m^{-(1-\delta)}$
at the expense of only a slightly weaker size lower bound, that is, of showing that the
protocol is of size at least $m^{(1-\delta)(\ell -1)}$ (instead of $m^{(1-\delta)}/2$).
Here $\delta$ is the same parameter as in~\cite{DBLP:conf/coco/RezendeV25}, that is,
$\delta$ such that  $m^\delta = 2K\cdot |\Sigma| \cdot 4\ell \cdot \log(2mN^2)$ for $K$
being the constant given from the Full Range Lemma. Note that in our case $\Sigma =
\{0,1\}$. Moreover, observe that $m^\delta \leq c' \ell \log(mN)$ for a large enough
absolute constant $c'$, and so, in particular, we will be proving the statement for
$\gamma  \geq 2 c' \ell \log(mN)/m \geq 2m^{-(1-\delta)}$.

Suppose that some small circuit has $1/2+\gamma$ correlation with
$f$.
This implies that the root rectangle $R=R^X \times R^Y$ obtained from the circuit 
has density at least
$\gamma$ in both sides (instead of density $1$, as in~\cite{DBLP:conf/coco/RezendeV25}).
Suppose that the DAG has size at most $m^{(1-\delta)(\ell -1)} \leq m^{(1-\delta)\ell} \cdot \gamma/2$, 
This changes the density of both the $X$-error and of the $Y$-error
    accumulated at the root to be at most $\gamma/2$ (instead of at most $1/2$),
and, therefore, the root rectangle, with errors removed, has density at
least $\gamma/2$ in both sides.
This only affects the argument that $ R^X \times (R^Y \setminus Y^R_{\err})$ is a $\star^N$-(pre-)structured rectangle, that is, that
$R^X$ is $\star^n$-predense and 
$|R^Y \setminus Y^R_{\err}|\geq
2^{mN-4\ell \log mN}$ (since in our case $\Sigma = \{0,1\}$). 

We first argue that $R^X$ is $\star^N$-predense. Indeed, for any $\emptyset \neq I \sseq X$ 
the min-entropy of $R^X_I$ 
is at least
\begin{align*}
    \card{I} \log m - \log(2/\gamma)
    &=
    \card{I} \log m - (1-\delta)\log m
    \\&\geq
    \card{I} \log m - (1-\delta) \card{I} \log m
    \\&\geq
    \delta \card{I} \log m.
\end{align*}
We have thus obtained the bound required to 
conclude that $R^X$ is $\star^N$-predense.

It remains to show that $|R^Y \setminus Y^R_{\err}|\geq  2^{mN-4\ell\log mN}$.
This follows easily since
 $|R^Y \setminus Y^R_{\err}|\geq 2^{mN} \cdot \gamma / 2 $
 and $ 
 \gamma/2 = m^{-(1-\delta)} \geq 2^{-\log m} \geq 2^{-4\ell\log mN}$. 
We therefore conclude that $ R^X \times (R^Y \setminus 
		Y^R_{\err})$ is a $\star^N$-(pre-)structured rectangle,
and the rest of the proof is the same.

This implies the size of lower bound of 
$m^{(1-\delta)(\ell-1)} \geq (m/c'\ell \log(mN))^{\ell-1}$ for any $\gamma \geq 2 c' \ell \log(mN)/m $. Choosing the absolute constant $c = 2c'$ gives us \myref{thm:constructive-lifting}.

\section{Proof of the Lower Bound for the Refutation Formula}
\label{sec:lower-bound}

In this appendix, we prove \myref{it:thm-extended-ref-unsat} of \myref{lem:extended-ref},
that is, we show that
if $F$ is unsatisfiable $3$-CNF 
formula over $n$ variables, with $ \poly(n)$ clauses, $d 
    \leq n$ and $t \leq 2^{\noverdepth}$,
and $\graph$ is a $(\blockslayer\cdot\pointersblock,\blockslayer,\delta,r,c)$-{bipartite expander graph},
then
$\eREF_d(F)$ requires Resolution proofs of width at least $2^{\Omega(r(c-1)/n)}$.
The proof follows along the same lines of the ones in~\cite{DBLP:conf/stoc/RezendeGNPR021}
and in~\cite{DBLP:conf/lagos/Rezende21}.

\subsection{Weak Graph \texorpdfstring{$\PHP$}{rPHP} Formula}
The formula we use is a graph version of the
\emph{retraction pigeonhole principle}, $\PHP$~\cite{Jerabek07,PudlakT19}.
The retraction pigeonhole principle 
was also used in~\cite{DBLP:conf/stoc/RezendeGNPR021} to prove a similar reduction.

In this variant of the pigeonhole principle, the pigeon-mapping is restricted to the edges 
of a given bipartite graph and, moreover, there is an efficient way to \emph{invert} the mapping. 
Formally, given a bipartite graph  $\graph = (U \disjointunion V, E)$ with $U=[m]$ and $V=[n]$, 
and an arbitrary family of surjections $\mathsf{sur}_i : [\delta] \rightarrow N(i)$ for each $i\in [m]$,
the variables of the formula $\PHP(\graph)$  describe two functions, $f\colon [m]\to[\delta]$ and $g\colon[n]\to[m]$,
as follows.
\begin{itemize}
	\item \emph{Pigeon map variables.} For every pigeon $i\in[m]$ there are variables $f_{ik}$, $k \in [\log \delta]$. The pigeon map variables encode in binary a hole $\mathsf{sur}_i(f(i)) \in N(i)$ that is expected to house pigeon $i$.
	\item \emph{Hole map variables.} For every hole $j\in [n]$ there are variables $g_{j\ell}$, $\ell\in [\log m]$. These variables encode in binary a pigeon $g(j)\in[m]$ that is expected to occupy hole $j$.
\end{itemize}
The axioms of $\PHP(G)$ state that for every $i\in[m]$ and $j\in N(i)$,
\begin{equation}
    \label{eq:invertible-gphp-functional-notation}
    \enspace \mathsf{sur}_i(f(i)) = j\enspace \Longrightarrow \enspace g(j) = i \,.
\end{equation}
Note that (\ref{eq:invertible-gphp-functional-notation}) corresponds to $\log m$ clauses
each of width $\delta$. Therefore,  
$\PHP(\graph)$ can be written as an $O(\delta)$-width CNF formula with $O(n\delta \log m)$ clauses.

It was shown in~\cite{BenSasson2001} that if $\graph$ is a good enough expander,
then the pigeonhole principle formula over $\graph$ requires large width to be refuted in resolution.
As argued in~\cite{DBLP:conf/coco/RezendeV25}, 
we can obtain the following, slightly stronger, lower bound on the width required to refute 
the pigeonhole principle formula over $\graph$, which in particular applies to our encoding $\rPHP(\graph)$.

\begin{theorem}[\cite{DBLP:conf/coco/RezendeV25}]\label{thm:widthlb}
    If $\graph$ is an $(m,n,\delta,r,c)$-bipartite expander graph then
    $$\frac{(c-1)r}{2} \leq 
    \wRes(\rPHP(\graph))
    \leq n+1 \,.$$
\end{theorem}

In particular, this implies that if 
$\graph$ is an $(m,n,\delta,r,\delta/2)$-bipartite expander graph from 
\myref{lem:random-expander-constant}
then $\wRes(\rPHP(\graph)) = \Omega(\sqrt{n})$.
And if $\graph=((U\cup V), E)$ is a complete bipartite graph 
with $|{U}| = m$, $|{V}| = n$,
then by \myref{rk:expander} $\wRes(\rPHP(\graph)) \geq  n/4$.

\subsection{Decision Tree Reductions}

Our goal is to show, via a moderately small-depth reduction,
that the width required for refuting 
$\myRef(F)$ is almost as large as the width required for 
refuting~$\PHP(\graph)$.
For the definition of reduction, %
we view a clause as a function, that is, given
a CNF formula $H$ over $m$ variables,
we view a clause $\clause\in H$ as a function from $\{0,1\}^m$ to $\{0,1\}$,
which can be computed by a depth-$|\clause|$ decision tree.
In this way, given a function $\redfun : \{0,1\}^n \rightarrow \{0,1\}^m$,
$\clause\circ \redfun$ is a function from $\{0,1\}^n$ to~$\{0,1\}$.
Moreover, if $\redfun_i$ is computed by a depth-$d$ decision tree,
then $\clause\circ \redfun$ can be naturally written as 
a ($d\cdot |\clause|$)-CNF formula 
(i.e., a CNF formula with width $d \cdot |\clause|$) over $n$ variables.

\begin{definition}[\cite{DBLP:conf/stoc/RezendeGNPR021}]\label{def:reduction}
Given two CNF formulas $F$ and $H$ over $n$ and $m$ variables %
respectively, there is a \emph{depth-$d$ reduction} from $F$ to $H$, denoted $F \leq_d H$, 
if there is a function $\redfun : \{0,1\}^n \rightarrow \{0,1\}^m$ such that
\begin{itemize}
\item each output bit of $\redfun_i : \{0,1\}^n \rightarrow \{0,1\}$ for $i\in [m]$ is computed 
by a depth-$d$ decision tree; and
\item for every axiom $\clause\in H$, each clause in the representation of $\clause\circ \redfun$ 
as a ($d\cdot |\clause|$)-CNF formula (as discussed above) is a weakening of some axiom of~$F$.
\end{itemize}
\end{definition}

The important property of reductions that we use is that width bounds carry over.
\begin{lemma}[\cite{DBLP:conf/stoc/RezendeGNPR021}]\label{lem:reduction}
If $F \leq_d H$, then $\wRes(F) \leq d \cdot \wRes(H)$.
\end{lemma}

For the rest of this section, we let 
$\graph$ be a $(\blockslayer\cdot\pointersblock,\blockslayer,\delta,r,c)$-{bipartite expander graph}, where $\blockslayer \leq \pointersblock$.
Combining \myref{thm:widthlb} and~\myref{lem:reduction} implies that,
in order to obtain the lower bound in \myref{lem:extended-ref}, 
it is enough to show that
\begin{equation} \label{eq:php-shref}
\PHP(\expandergraph)~\leq_{O(n)}~\myRef(F) \,.
\end{equation}

\subsection{Overview of the Reduction}

The proof of the reduction is similar to that of~\cite{DBLP:conf/lagos/Rezende21} which in turn follows the one in~\cite{DBLP:conf/stoc/RezendeGNPR021}.
As in~\cite{DBLP:conf/lagos/Rezende21}, the reduction is guided by a full tree $\calt$ of
height
$\formuladepth$ and arity~$\pointersblock$.
Recall that we consider a fixed ordering %
of the variables of $F$. Each node of~$\calT$ contains a clause.
The root contains the empty clause ($\bot$) and
each node at distance $\alpha$ from the root
contains a clause with $\alpha n/d$
defined recursively as follows.
If a non-leaf node at distance $\alpha$ from the root 
contains a clause~$C$ using the first
$\alpha \noverdepth$ variables, 
each of its $\pointersblock$ children
contains one of the possible
weakenings of $C$ to the next $\noverdepth$
variables, that is, for $\sigma\in [\pointersblock]$,
the $\sigma$-th child
contains the clause
$C \cup \{x^{\sigma_0}_{\alpha\noverdepth }, x^{\sigma_1}_{\alpha\noverdepth + 1}, \ldots,x^{\sigma_{\noverdepth-1}}_{\alpha\noverdepth + \noverdepth - 1}\}$ (recall that $x_i^b = \neg x_i$ if $b=0$ and $x_i$ otherwise).
Thus, $\calT$ has $2^n$ leaves,
each containing a different clause of width $n$.

\subsection{The Reduction Function}

We start by describing the decision trees that compute the
variables of~$\myRef(F)$. 
For simplicity, we sometimes describe 
one decision tree that computes several variables at once.

\paragraph{Pointer variables.} 
The pointer variables of $\eREF_d$
are defined in terms of the $f$-variables 
of $\PHP(\expandergraph)$. 
Since each block has $\pointersblock$ pointers,
each layer $\alpha\in \set{0, \dots, d-1}$ has 
a total of $\blockslayer \cdot\pointersblock$ pointers 
that should
be mapped to the $\blockslayer$ blocks at level $\alpha +1$.
The idea is to define this mapping by ``copying''
the $\PHP(\expandergraph)$ mapping from the $\blockslayer\cdot\pointersblock$
pigeons to the $\blockslayer$ holes. %
Formally, 
let $B$ be an internal block,
let~$z\in \zeroset{\pointersblock}$ and let
$i = \col(B) \cdot\pointersblock + z$.
For all $k\in [\delta]$, we define the pointer variables 
$\point\bin^B_{z,k} = f_{i,k}$
and for all $B'$ %
we define $\point^B_{B',z} = 1$ if and only if $\col(B') = \sur_i(f(i))$.
This can clearly be done by a decision tree of depth~$\log \delta = O(n)$.

\paragraph{Enabled and literal variables.}
We have now reached the most complicated part of the reduction.
Consider the function $g'\colon[\blockslayer]\to[\blockslayer\cdot\pointersblock]\cup \{\star\}$
defined as
\begin{equation}
g'(j) = \left\{ \begin{array}{ll} 
g(j) & \text{ if } \sur_{g(j)}(f(g(j))) = j\,,  \\
\star & \text{ otherwise.}
\end{array}\right.
\end{equation}
In order to determine whether a block $B$
is enabled or not and
what literal set it contains, 
we recursively define a path of blocks and pointers starting from~$B$,
denoted $\pathB{B}$, 
in the following way.
Let $j=\col(B)$.
If $B$ is the root or if $g'(j) = \star$, 
then the path ends at $B$, that is, $\pathB{B} = B$.
Otherwise, let 
\begin{equation}
j' = \left\lfloor{\frac{g'(j)}{\blockslayer}}\right\rfloor 
\hspace{50pt}\text{and}\hspace{50pt}
\sigma = g'(j) \mod \blockslayer
\eqperiod
\end{equation}
Note that $j'$ and $\sigma$ are defined so that 
the $\sigma$th pointer of the block $B'$ with $\col(B') = j'$ and $\layer(B') = \alpha - 1$
is $j$.
We define $\pathB{B}$ to be
$(B, \sigma)$ concatenated with $\pathB{B'}$.

We are now ready to define a decision tree that computes the 
enabled variable and 
the literal variables of the block $B$. In order to do so,
the decision tree must compute $\pathB{B}$, which can be done in
depth~$O(\formuladepth (\log (\blockslayer \cdot \pointersblock) + \log \delta)) = O(n)$.
Indeed, the length of $\pathB{B}$ is at most $\formuladepth $ 
and to determine the next block in a path it is enough to compute $g'$,
which can be done with 
$O(\log (\blockslayer \cdot \pointersblock))=O(\noverdepth)$ queries to the $g$-variables and $O(\log \delta) \le O(n/d)$ queries to the $f$-variables
of $\PHP(\expandergraph)$.

Once $\pathB{B}$ is determined, the decision tree outputs
the value of the variables as follows.
If $\pathB{B}$ does not end at the root,
the enabled variable 
$\enable^{B} = 0$, the literal variables $\lit^{B}_{x_i} = 0$
for $\ell \in \Lit$ (if B has $\lit$ variables), and the summary variables $\summlit_{i,z}^B$
are all set to $1$.
If $\pathB{B}$ ends at the root,
let $\pi$ be
the reverse of $\pathB{B}$. 
Note that $\pi$ defines a unique path
in the tree~$\calT$ starting from the root and following the 
pointers in $\pi$. This path in~$\calT$ ends in some block
$B'$ at level $\alpha$. 
The decision tree sets the enabled variable 
$\enable^{B}$ to $1$ and sets the literal variables $\lit^{B}_\ell$
for $\ell \in \Lit$
to match the assignment of the block $B'$.

\paragraph{Weakening variables.}
The weakening variables are computed ``honestly'' given the assignments
contained at the blocks at the last layer. 
That is, in order to determine
the weakening variables for the block $B$,
the decision tree computes $\pathB{B}$ (in depth $O(n)$),
in this way determining the literal variables of $B$,
and then it outputs the value of each 
weakening variable $\weak$ and $\weak\bin$
so that this block points to an axiom $A$ of $F$
which it is the weakening of. Since there are $n$
literals in $B$ and $F$ is unsatisfiable, such an $A$ must exist.
We note that, if $\pathB{B}$ does not end at the root, then
the literal variables and the enabled variable 
of $B$ are all set to $0$ and the
weakening variables can point to an arbitrary axiom.

\vspace{1em}

Therefore, we conclude 
that the reduction $r$ satisfies the required condition 
that for each~$i$ there is a decision tree of depth $O(n)$ that
computes $r_i$.

\subsection{Correctness of the Reduction} %
We now show that the reduction satisfies the
second condition of \myref{def:reduction}. Recall that $\redfun$ is the function
defined by the decision trees described above.
Let $D$ be an axiom of $\myRef(F)$. It is enough to show that
for every partial assignment $\rho$ to the variables of $\PHP(\expandergraph)$,
if $\rho$ falsifies $D\circ r$ then $\rho$ falsifies an axiom of $\PHP(\expandergraph)$.

We first note that, by definition of $\redfun$, if $D$ is the root axiom \ref{axiom:root-enabled}; a binary pointer axiom \ref{axiom:point-binary-pointer}; a leaf axiom \ref{axiom:must-appear-after-weak}, \ref{axiom:weaken-binary-pointer}, or \ref{axiom:summ-to-lits}; or a backchecking axiom \ref{axiom:pass-on-summ-lits} or \ref{axiom:pointer-to-summ-lits}, then
$D\circ \redfun$ is a tautology, that is, for any full 
assignment $\rho$ to the variables of $\PHP(\expandergraph)$, the formula $D\circ \redfun$ is satisfied.

It remains to consider the case when $D$ is one of the 
 \ref{axiom:must-enable}, \ref{axiom:pass-on-summ-lits}, or \ref{axiom:pointer-to-summ-lits} axioms,
 that is $( \enable^B \land  \point^B_{B',z} ) \to \enable^{B'} $, $(\enable^B \land \point^B_{B',z'} \land \summlit_{i',z}^{B} ) 
            \to 
            \summlit_{i',z}^{B'}  $
            or $( \enable^B \land \point^B_{B',z}) \to \summlit_{i+1,z}^{B'}$, respectively, for some $B,B' \in [s]$ with $i' \le i = \layer(B) = \layer(B') - 1$.
In all three cases, for these clauses to be falsified,  both $\enable^B$ and $\point^B_{B',z}$ must be set to $1$.
Note that $\enable^B$ will only be set to $1$ if
$\pathB{B}$ ends at the root and
$\point^B_{B',z}$ will only be set to $1$ if $\col(B') = \sur_j (f(j))$ for $j = \col(B) \cdot 2^{\noverdepth} + z$. 
 
Consider when $D$ is one of the  \ref{axiom:must-enable} axioms.
Note that $\enable^{B'}$ will only be set to $0$ in the reduction if $\pathB{B'}$ does not end at the root.
This implies that the value returned for $g(\col(B'))$ was something other than $j$ (corresponding to the $z$th pointer of $B$), since otherwise $g'(\col(B'))$ would be equal to $g(\col(B'))$ and $\pathB{B'}$ would be $\pathB{B}$ concatenated with $(B', B)$ and would, therefore, reach the root.
Since $\col(B') = \sur_j (f(j))$, the clause
$\sur_j (f(j)) = \col(B') \Rightarrow g(\col(B')) = j$ from $\rPHP(G)$ is violated.

Now consider the case when $D$ is one of the \ref{axiom:pass-on-summ-lits} or \ref{axiom:pointer-to-summ-lits} axioms.
If the value returned for $g(\col(B'))$ is $j$ (corresponding to the $z$th pointer of $B$), then 
$\pathB{B'}$ would go through $B$ via the $z$th pointer of $B$.
This would imply that, in the case of \ref{axiom:pass-on-summ-lits} both $\summlit$ variables would be the same and thus \ref{axiom:pass-on-summ-lits} would not be falsified; and in the case of \ref{axiom:pointer-to-summ-lits} $\summlit_{i+1,z}^{B'}$ would be $1$, because $\point^B_{B',z}$ is $1$ and $\layer(B) = i$,
and thus \ref{axiom:pointer-to-summ-lits} would not be falsified.
So it must be the case that $g(\col(B'))$ is not $j$, and therefore
we have that
the clause $\sur_{j} (f(j)) = \col(B') \Rightarrow g(\col(B')) = j$ from $\rPHP(G)$ is violated.

\end{document}